\documentclass[12pt]{article}

\usepackage{bbm}
\usepackage{ifpdf}
\usepackage{amsmath}
\usepackage{amsthm}
\usepackage{url}
\usepackage{amsfonts}
\usepackage{alltt, amssymb}
\usepackage{hyperref}
\usepackage{mathrsfs,stmaryrd}
\usepackage{graphicx}
\usepackage{color}
\usepackage{fullpage}

\def\CC {\ensuremath{\mathsf{C}}}
\def\EE {\ensuremath{\mathsf{E}}}
\def\Q {\ensuremath{\mathbb{Q}}}
\def\N {\ensuremath{\mathbb{N}}}

\def\Z {\ensuremath{\mathbb{Z}}}
\def\F {\ensuremath{\mathbb{F}}}

\def\K {\ensuremath{\mathbb{K}}}
\def\Kbar {\ensuremath{\overline{\mathbb{K}}}}
\def\A {\ensuremath{\mathbb{A}}}
\def\M {\ensuremath{\mathsf{M}}}
\def\bd {\ensuremath{\mathbf{d}}}

\def\Tt {\ensuremath{\mathbf{T}}}
\def\Ss {\ensuremath{\mathbf{S}}}
\def\Uu {\ensuremath{\mathbf{U}}}
\def\Vv {\ensuremath{\mathbf{V}}}

\def\x {\ensuremath{\mathbf{x}}}
\def\y {\ensuremath{\mathbf{y}}}

\def\X {\ensuremath{\mathbf{X}}}

\def\Ur {\ensuremath{\mathscr U}}
\def\Wr {\ensuremath{\mathscr W}}
\def\Vr {\ensuremath{\mathscr V}}

\def\Dr {\ensuremath{\mathscr D}}
\def\Dec {\ensuremath{\mathsf{Dec}}}

\def\Ot {O\tilde{~}}

\newtheorem{Theo}{Theorem}
\newtheorem{Prop}{Proposition} 
\newtheorem{Lemma}{Lemma}

\title{On the complexity of computing with zero-dimensional triangular sets}

\author{
  \begin{tabular}{cc}
    Adrien Poteaux$^{\star}$ & \'Eric Schost$^\dagger$ \\ 
    \texttt{adrien.poteaux@lip6.fr} & \texttt{eschost@uwo.ca}\\
  \end{tabular}\\
  ${}^\dagger$: {\scriptsize Computer Science Department, 
    The University of Western Ontario, London, ON, Canada}\\
  ${}^\star$: {\scriptsize UPMC, Univ Paris 06, INRIA, Paris-Rocquencourt center, SALSA Project, LIP6/CNRS UMR 7606 France}
}

\begin{document}

\maketitle

%% no citation in the abstract
\begin{abstract}
  We study the complexity of some fundamental operations for
  triangular sets in dimension zero. Using Las-Vegas algorithms, we
  prove that one can perform such operations as change of order,
  equiprojectable decomposition, or quasi-inverse computation with a
    cost that is essentially that of {\em modular composition}. Over
    an abstract field, this leads to a subquadratic cost (with respect
    to the degree of the underlying algebraic set). Over a finite
    field, in a boolean RAM model, we obtain a quasi-linear running
    time using Kedlaya and Umans' algorithm for modular composition.

  Conversely, we also show how to reduce the problem of modular
  composition to change of order for triangular sets, so that all
  these problems are essentially equivalent.

  Our algorithms are implemented in Maple; we present some
  experimental results.
\end{abstract}

%%%%%%%%%%%%%%%%%%%%%%%%%%%%
% à virer plus tard (ou adapter si on veut garder le toc)
%% \setcounter{tocdepth}{3}
%% \tableofcontents
%% \setcounter{page}{1}
%%%%%%%%%%%%%%%%%%%%%%%%%%%%

%%%%%%%%%%%%%%%%%%%%%%%%%%%%%%%%%%%%%%%%%%%%%%%%%%%%%%%%%%%%
%%%%%%%%%%%%%%%%%%%%%%%%%%%%%%%%%%%%%%%%%%%%%%%%%%%%%%%%%%%%
%%%%%%%%%%%%%%%%%%%%%%%%%%%%%%%%%%%%%%%%%%%%%%%%%%%%%%%%%%%%

\section{Introduction}\label{sec:intro}

Triangular sets (in dimension zero, in this paper) are families of
polynomials with a simple triangular structure, which turns out to be
well adapted to solve many problems for systems of polynomial
equations. As a result, there is now a vast literature dedicated to
algorithms with triangular sets, their generalization to {\em regular
  chains}, and applications: without being exhaustive, we refer the
reader
to~\cite{Kalkbrener93,AuLaMo99,MorenoMaza00,Hubert03,Schost03,Schost03b}.

However, from the algorithmic point of view, many questions
remain. Despite a growing amount of work~\cite{LiMoSc09,
  LiMoPa09,BoChHoSc09}, the complexity of many basic operations with
triangular sets (such as set-theoretic operations on their zero-sets,
change of variable order, or arithmetic operations modulo a triangular
set) remains imperfectly understood.

The aim of this paper is to answer some of these questions, by
describing fast algorithms for several operations with triangular
sets, extending our previous results from~\cite{PoSc10}. In
particular, we will focus on the relationship between these problems
and some classical operations on univariate and bivariate polynomials,
called {\em modular composition} and {\em power projection}. To
describe these issues with more details, we need a few definitions.

%%%%%%%%%%%%%%%%%%%%%%%%%%%%%%%%%%%%%%%%%%%%%%%%%%%%%%%%%%%%

\subsection{Basic definitions}\label{ssec:equiproj}

\paragraph{Triangular sets.} Let $\K$ be our base field, and let
$\X=X_1,\dots,X_n$ be indeterminates over $\K$; we order them as $X_1
< \cdots < X_n$.  A (monic) triangular set $\Tt=(T_1,\dots,T_n)$, for
this variable order, is a family of polynomials in $\K[\X]$ with the
following triangular structure
\[\Tt \left | \begin{array}{l}
    T_n(X_1,\dots,X_n)\\
    ~~\vdots\\
    T_1(X_1), 
  \end{array}\right .\]
and such that for all~$i$, $T_i$ is monic in $X_i$ and reduced modulo
$\langle T_1,\dots,T_{i-1}\rangle$, in the sense that $\deg(T_i,X_j) <
\deg(T_j,X_j)$ for $j<i$; in particular, $\Tt$ is a zero-dimensional
Gr\"obner basis for the lexicographic order induced by $X_1 < \cdots <
X_n$. In all that follows, we will impose the condition that $\K$ is a
perfect field; often, we will also require that $\langle \Tt\rangle$
is a radical ideal.

We write $d_i=\deg(T_i,X_i)$; $\bd=(d_1,\dots,d_n)$ will be called the
{\em multidegree} of $\Tt$. Define further $R_\Tt=\K[\X]/\langle
\Tt\rangle$. Then, $\delta_\Tt=d_1 \cdots d_n$ is the natural
complexity measure associated to computations modulo $\langle \Tt
\rangle$, as it represents the dimension of the residue class ring
$R_\Tt$ over~$\K$. This integer will be called the {\em degree} of
$\Tt$.

In all our algorithms, elements of $R_\Tt$ are represented on the
monomial basis $B_\Tt=\{X_1^{a_1}\cdots X_n^{a_n} \ | \ 0 \le a_i <
d_i \text{~for all $i$}\}$. Dually, all $\K$-linear forms $R_\Tt \to
\K$ are represented by their values on the basis $B_\Tt$.

\paragraph{Equiprojectable sets.} Not every zero-dimensional radical ideal $I$ in
$\K[\X]$ admits a triangular set of generators: this is the case only
when the zero-set $V=V(I) \subset \overline \K^n$ possesses a
geometric property called {\em equiprojectability}~\cite{AuVa00}. For
the moment, we will simply give an idea of the definition; proper
definitions are in Section~\ref{ssec:equi}.

Roughly speaking, $V$ is equiprojectable if all fibers of the
projection $V\to\Kbar^{n-1}$ have the same cardinality, and similarly
for the further projections to $\Kbar^{n-2},\dots,\Kbar$.  For
instance, of the following pictures, the left-hand one describes an
equiprojectable set, whereas the right-hand one does not (since the
rightmost fiber has a larger cardinality than the others).

\begin{center}
\begin{picture}(0,0)%
\includegraphics{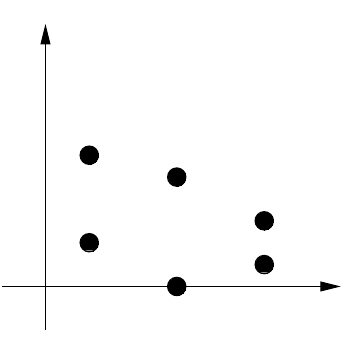}%
\end{picture}%
\setlength{\unitlength}{2763sp}%
\begingroup\makeatletter\ifx\SetFigFont\undefined%
\gdef\SetFigFont#1#2#3#4#5{%
  \reset@font\fontsize{#1}{#2pt}%
  \fontfamily{#3}\fontseries{#4}\fontshape{#5}%
  \selectfont}%
\fi\endgroup%
\begin{picture}(2427,2328)(589,-1630)
\put(3001,-1561){\makebox(0,0)[lb]{\smash{{\SetFigFont{8}{9.6}{\rmdefault}{\mddefault}{\updefault}{\color[rgb]{0,0,0}$X_1$}%
}}}}
\put(901,539){\makebox(0,0)[lb]{\smash{{\SetFigFont{8}{9.6}{\rmdefault}{\mddefault}{\updefault}{\color[rgb]{0,0,0}$X_2$}%
}}}}
\end{picture}%

\hspace{1cm} 
\begin{picture}(0,0)%
\includegraphics{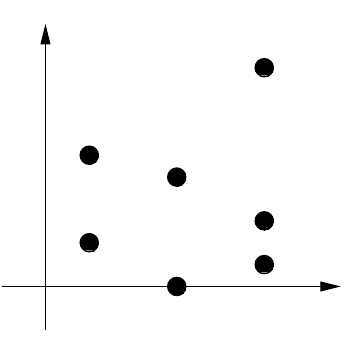}%
\end{picture}%
\setlength{\unitlength}{2763sp}%
\begingroup\makeatletter\ifx\SetFigFont\undefined%
\gdef\SetFigFont#1#2#3#4#5{%
  \reset@font\fontsize{#1}{#2pt}%
  \fontfamily{#3}\fontseries{#4}\fontshape{#5}%
  \selectfont}%
\fi\endgroup%
\begin{picture}(2427,2328)(589,-1630)
\put(3001,-1561){\makebox(0,0)[lb]{\smash{{\SetFigFont{8}{9.6}{\rmdefault}{\mddefault}{\updefault}{\color[rgb]{0,0,0}$X_1$}%
}}}}
\put(901,539){\makebox(0,0)[lb]{\smash{{\SetFigFont{8}{9.6}{\rmdefault}{\mddefault}{\updefault}{\color[rgb]{0,0,0}$X_2$}%
}}}}
\end{picture}%

\end{center}

The relationship with triangular representations is described
in~\cite{AuVa00}: $V$ is equiprojectable if and only if its defining
ideal $I$ is generated by a triangular set (for this equivalence, it
is required that the base field be perfect).

\paragraph{Equiprojectable decomposition.} Any finite set can be
decomposed, in general not uniquely, into a finite union of pairwise
disjoint equiprojectable sets. At the level of ideals, this amounts to
write a zero-dimensional radical ideal $I$ as $I=\langle \Tt^{(1)}
\rangle \cap \cdots \cap \langle \Tt^{(s)} \rangle$, with all
$\Tt^{(j)}$ being triangular sets and all ideals $\langle \Tt^{(j)}
\rangle$ being pairwise coprime. Of course, starting from $I$ in
$\K[\X]$, we want all $\Tt^{(j)}$ to have coefficients in $\K$ as
well.

To solve the non-uniqueness issue, the decomposition of $I$ into an
intersection of {\em maximal} ideals may appear as a good candidate;
however, it suffers from significant drawbacks. For instance,
computing it requires us to factor polynomials over $\K$, or
extensions of it: even if we strengthen our model by requiring that
$\K$ and its finite extensions support this operation, it is usually
prohibitively costly.

There exists another canonical way to find such a decomposition,
called the {\em equiprojectable
  decomposition}~\cite{DaMoScWuXi05}. For instance, among its useful
properties is the fact that it behaves well under specialization: if
$\K$ is the fraction field of a ring $\A$ such as
$\A=k[Z_1,\dots,Z_r]$ or $\A=\Z$ and $\mathfrak{m}$ is a maximal ideal
of $\A$, the equiprojectable decomposition of $(I \bmod \mathfrak{m})$
coincides with the equiprojectable decomposition of $I$, reduced
modulo $\mathfrak{m}$, for ``most'' maximal ideals~$\mathfrak{m}$. We
refer to~\cite{DaMoScWuXi05} for more precise statements; here, we
simply point out that this property makes it for instance possible to
apply modular methods, such as Hensel lifting
techniques~\cite{Schost03,Schost03b}, to recover the equiprojectable
decomposition of $I$ starting from that of $(I \bmod \mathfrak{m})$;
the decomposition of $I$ into maximal ideals does not have this useful
specialization property.

While the definition of the equiprojectable decomposition is
technical, the idea is simple. We will proceed geometrically: to
obtain the equiprojectable decomposition of a finite set $V \subset
\overline \K^n$, we first split it using the cardinality of the fibers
of the projection $\Kbar^n \to \Kbar^{n-1}$. Then we apply the same
process to all the components we obtained, using the projection to
$\Kbar^{n-2}$, and so on (again, we refer the reader to
Section~\ref{ssec:equi} for precise definitions).  The following
picture (from~\cite{DaMoScWuXi05}) shows the equiprojectable
decomposition of the non-equiprojectable set $V$ of the former example.

\begin{center}
\begin{picture}(0,0)%
\includegraphics{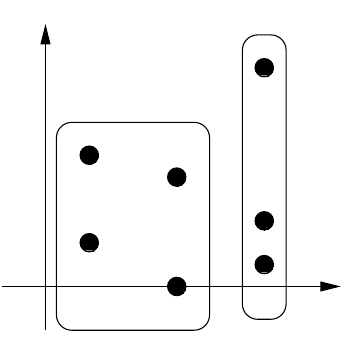}%
\end{picture}%
\setlength{\unitlength}{2763sp}%
\begingroup\makeatletter\ifx\SetFigFont\undefined%
\gdef\SetFigFont#1#2#3#4#5{%
  \reset@font\fontsize{#1}{#2pt}%
  \fontfamily{#3}\fontseries{#4}\fontshape{#5}%
  \selectfont}%
\fi\endgroup%
\begin{picture}(2427,2328)(589,-1630)
\put(3001,-1561){\makebox(0,0)[lb]{\smash{{\SetFigFont{8}{9.6}{\rmdefault}{\mddefault}{\updefault}{\color[rgb]{0,0,0}$X_1$}%
}}}}
\put(901,539){\makebox(0,0)[lb]{\smash{{\SetFigFont{8}{9.6}{\rmdefault}{\mddefault}{\updefault}{\color[rgb]{0,0,0}$X_2$}%
}}}}
\end{picture}%

\end{center}

%% \centerline{\includegraphics[width=4cm]{fig1}}

Each component of the equiprojectable decomposition is an
equiprojectable set. As a result, this construction allows us to
represent an arbitrary finite set $V$, defined over $\K$, by means of
a canonical family of triangular sets with coefficients in $\K$, that
depends only on the order $<$ we have chosen on the variables. The
collection of these triangular sets will thus be denoted by
$\Dr(V,<)$.

%%%%%%%%%%%%%%%%%%%%%%%%%%%%%%%%%%%%%%%%%%%%%%%%%%%%%%%%%%%%

\subsection{Our contribution}

Our purpose is to give algorithms for various operations involving a
triangular set, or a family thereof. We will make these questions more
precise below; for the moment, one should have in mind problems such
as modular arithmetic, computation of the equiprojectable
decomposition, or change of order on the variables.

\paragraph{Two central problems.}
The following two problems, called modular composition and power
projection, will be at the heart of our algorithms.  Given a
triangular set $\Tt$ in $\K[X_1,\dots,X_n]$, the general forms of
these questions are the following.
\begin{itemize}
\item {\em modular composition:} given $F$ in $\K[Y_1,\dots,Y_m]$,
  with $\deg(F,Y_i) < f_i$ for all $i$, and $(G_1,\dots,G_m)$ in
  $R_\Tt^m$, compute $F(G_1,\dots,G_m) \in R_\Tt$
\item {\em power projection:} given a linear form $\ell:R_\Tt \to \K$,
  $(G_1,\dots,G_m)$ in $R_\Tt^m$ and bounds $f_1,\dots,f_m$, compute
  $\ell(G_1^{c_1}\cdots G_m^{c_m})$, for all $c_1 < f_1,\dots,c_m <
  f_m$.
\end{itemize}
In both cases, we will write ${\bf f}=(f_1,\dots,f_m)$ and
$\delta_{\bf f}=f_1\cdots f_m$, so that the size of the problem is
characterized by $\delta_{\bf f}$ and $\delta_\Tt$. We will call
$(m,n)$ the {\em parameters} for these questions, and
$\max(\delta_{\bf f},\delta_\Tt)$ the {\em size}. When $\Tt$ and
$G_1,\dots,G_m$ are fixed, the two problems become linear in
respectively $F$ and $\ell$; as it turns out, they are dual problems,
as was observed by Shoup for $m=n=1$~\cite{Shoup94}.

The only cases we will need actually have parameters $(m,n)$ in
$\{1,2\}$.  Besides, we will always suppose that $\delta_{\bf f} \le
\delta_\Tt$, so that all costs can be measured in terms of
$\delta_\Tt$ only. However, even in this simple situation, these
questions have resisted many attempts.

As of now, no quasi-linear time algorithm is known in an algebraic
complexity model (say using an algebraic RAM, counting field
operations at unit cost). Among the best results known to us is that
both operations can be done in time $O(\delta_\Tt^{(\omega+1)/2})$,
where $\omega$ is such that matrices over $\K$ of size $n$ can be
multiplied in time $O(n^\omega)$; we assume $\omega > 2$, otherwise
logarithmic terms may appear. Using the exponent $\omega \le 2.38$
from~\cite{CoWi90}, this gives the subquadratic estimate
$O(\delta_\Tt^{1.69})$.

For $(m,n)=(1,1)$, this claim follows from respectively Brent and
Kung's modular composition algorithm~\cite{BrKu78} and Shoup's power
projection algorithm~\cite{Shoup94}, which is actually the transpose
of Brent and Kung's. For power projection, extensions to parameters
$(m,n)=(1,2)$ are in~\cite{Shoup99,Kaltofen00,BoFlSaSc06}, and the
case $(m,n)=(2,2)$ is partially dealt with in~\cite{PaSc06}. For
completeness, in Section~\ref{ssec:basics}, we will give
straightforward extensions of the Brent-Kung and Shoup algorithms to
all cases $(m,n) \in \{1,2\}$, establishing the bound
$O(\delta_\Tt^{(\omega+1)/2})$ claimed above.

We will thus write $\CC:\N \to \N$ to denote a function such that over
any field, one can do both modular composition and power projection in
$\CC(\delta_\Tt)$ base field operations, under the assumptions that
the parameters $(m,n)$ are in $\{1,2\}$ and $\delta_{\bf f} \le
\delta_\Tt$.  We take $\CC$ super-linear, in the sense that we require
that $\CC(d_1 + d_2) \ge \CC(d_1)+\CC(d_2)$ holds for all
$d_1,d_2$. Then, the former discussion shows that we can take $\CC(d)
\in O(d^{(\omega+1)/2}) \subset O(d^{1.69})$.

Some further restrictions are imposed on the function $\CC$. As is now
customary, we let $\M:\N \to \N$ be such that over any ring,
polynomials of degree less than $d$ can be multiplied in $\M(d)$ base
ring operations; we make the standard superlinearity assumptions
of~\cite[Chapter~8]{GaGe03}. Using Cantor and Kaltofen's
algorithm~\cite{CaKa91}, we can take $\M(d)$ in $O(d \log(d)
\log\log(d))$. Then, to simplify several estimates, we also make the
reasonable assumption that $\M(d)\log(d)$ is in $O(\CC(d))$; this is
the case for $\M(d)$ quasi-linear and $\CC(d)=d^{(\omega+1)/2}$.

\paragraph{The Kedlaya-Umans algorithm and its applications.}
In a boolean model (using a boolean RAM, with logarithmic cost for
data access), and for $\K=\F_q$, it turns out that one can do much
better than in the algebraic model for modular composition and power
projection.

The best known result comes from Kedlaya and Umans'
work~\cite{KeUm10}: for $n=1$, they show how to solve both problems in
$\delta_\Tt^{1+\varepsilon}\log(q)^{1+o(1)}$ bit operations, for all
$\varepsilon > 0$. Their algorithm uses modular techniques
(transferring the problem over $\F_q$ to a problem over $\Z$, and vice
versa), and the idea does not seem to extend easily to an arbitrary
base field. In~\cite{PoSc10}, we described an extension of this result
to any parameters $(m,n) \in \{1,2\}$, with a running time of
$\delta_\Tt^{1+\varepsilon}\Ot(\log(q))$ bit operations for any
$\varepsilon > 0$; the $\Ot$ notation indicates the omission of
polylogarithmic factors of the form $\log\log(q)^{O(1)}$.

In this paper, we will be interested in both models, algebraic and
boolean. Now, for a given algorithm, the cost analysis in the boolean
model differs from the analysis in the algebraic model (where we only
count base field operations) by a few aspects. A minor issue is that
we should count the cost of fetching data (which grows like $\log(a)$,
to access the contents at address $a$). Another difference is that in
the boolean model, we need to take into account the boolean cost of
operations in $\F_q$: disregarding the cost of fetching data, any
arithmetic operations in $\F_q$ can be done in $O\tilde{~}(\log(q))$
bit operations, say $\log(q)\log\log(q)^k$ for some fixed $k\ge 0$.

As a result, in what follows, in all rigor, we should prove most
statements twice, once in the algebraic complexity model and once in
the boolean one. To avoid making the paper excessively heavy, we will
indeed state our main results twice, but {\em all intermediate results
  and proofs will be given for the algebraic model}. There would
actually be no major difference in the boolean model, only some extra
bookkeeping, on the basis of the remarks in the previous paragraph.

Similarly to the algebraic case, $\CC_{\rm bool}$ will thus denote a
function such that one can do both modular composition and power
projection over $\F_q$ using $\CC_{\rm bool}(\delta_\Tt,q)$ bit
operations, assuming that the parameters $(m,n)$ are in
$\{1,2\}$ and that $\delta_{\bf f} \le \delta_\Tt$. As before, we require
that $\CC_{\rm bool}(d_1 + d_2,q) \ge \CC_{\rm bool}(d_1,q)+\CC_{\rm
  bool}(d_2,q)$ holds for all $d_1,d_2,q$. As in the algebraic case,
we will also assume that the cost of polynomial multiplication and
related operations can be absorbed into $\CC_{\rm bool}$: explicitly,
we require that for any function $f(d) \in \Ot(d)$, the function
$f(d)\log(q)\log\log(q)^k$ is in $O(\CC_{\rm bool}(d,q))$, where $k$
is the constant introduced above. The results of~\cite{PoSc10} imply
that we can take $\CC_{\rm bool}(d,q)$ in
$d^{1+\varepsilon}\Ot(\log(q))$ for any $\varepsilon>0$.

\paragraph{Main results.}
The questions we will consider are the following set-theoretic
operations. In all the following items, {\em all triangular sets are
  supposed to generate zero-dimensional radical ideals}.
\begin{itemize}
\item [${\bf P}_1.$] Given triangular sets
  $\Tt^{(1)},\dots,\Tt^{(\ell)}$ and $\Ss^{(1)},\dots,\Ss^{(r)}$ in
  $\K[X_1,\dots,X_n]$, for a variable order $<$, and given a target
  variable order $<'$, compute the equiprojectable decomposition
  $$\Dr \big (V(\Tt^{(1)}) \cup \cdots \cup V(\Tt^{(\ell)}) -
  V(\Ss^{(1)}) - \cdots - V(\Ss^{(r)}),<'\big).$$ We let
  $\delta_1$ be the sum of the degrees of
  $\Tt^{(1)},\dots,\Tt^{(\ell)}$ and $\Ss^{(1)},\dots,\Ss^{(r)}$.
\item [${\bf P}_2.$] Given a triangular set $\Tt$ in
  $\K[X_1,\dots,X_n]$, for a variable order $<$, as well as $F$ in
  $R_\Tt$ and a target variable order $<'$, compute the
  equiprojectable decompositions
  $$\Dr(V(\Tt)\cap V(F),<') \quad\text{and}\quad \Dr(V(\Tt)-
  V(F),<');$$ for every $\Tt'$ in $\Dr(V(\Tt)- V(F),<')$, compute also the
  inverse of $F$ in $R_{\Tt'}$. (Note that even if $F$ is only defined
  modulo $\langle \Tt \rangle$, the two sets above are actually
  defined unambiguously.)  In this case, we let $\delta_2$ be the
  degree of $\Tt$.
\end{itemize}
These questions are general enough to allow us to solve a variety of
classical problems for triangular sets. When the initial and target
orders are the same, and when $r=0$, the first question amounts to
compute the equiprojectable decomposition of a family of triangular
sets, which is a key subroutine in the algorithms
of~\cite{DaMoScWuXi05}. When the initial and target orders are
different, taking only a single triangular set $\Tt$ as input, the
first question allows us to perform a change of order on $\Tt$, and to
output a canonical family of triangular sets for the target order.
Taking the same order for input and output, the second operation
allows us to compute the {\em quasi-inverse} of a polynomial $F$
modulo $\langle \Tt\rangle$, which amounts to split $V(\Tt)$ into its
components where $F$ vanishes, resp. is invertible.  This is an
important subroutine for triangular decomposition
algorithms~\cite{LiMoPa09}.

With that being said, our first main results are the following:
\begin{Theo}
  \label{theo:CtoE}
  In an algebraic RAM complexity model, the following holds over any
  field $\K$ of characteristic $p$:
  \begin{itemize}
  \item if $p=0$ or $p$ is greater than $\delta_1^2$, one can answer
    question ${\bf P}_1$ using an expected
    $O(n\CC(\delta_1)(n+\log(\delta_1)))$ base field
    operations;
  \item if $p=0$ or $p$ is greater than $\delta_2^2$, one can answer
    question ${\bf P}_2$ using an expected
    $O(n\CC(\delta_2)(n+\log(\delta_2)))$ base field operations.
  \end{itemize}
  In a boolean RAM complexity model, the following holds over any
  finite field $\F_q$ of characteristic $p$:
  \begin{itemize}
  \item if $p$ is greater than $\delta_1^2$, one can answer question
    ${\bf P}_1$ using an expected $O(n\CC_{\rm
      bool}(\delta_1,q)(n+\log(\delta_1)))$ bit operations;
  \item if $p$ is greater than $\delta_2^2$, one can answer question
    ${\bf P}_2$ using an expected $O(n\CC_{\rm
      bool}(\delta_2,q)(n+\log(\delta_2)))$ bit operations.
  \end{itemize}
\end{Theo}
Using the estimates of the previous paragraphs, the former costs are
$O\tilde{~}(n^2 \delta_1^{(\omega+1)/2})$ and $O\tilde{~}(n^2
\delta_2^{(\omega+1)/2})$, and the latter are
$n^2\delta_1^{1+\varepsilon}\Ot(\log(q))$ and $n^2
\delta_2^{1+\varepsilon}\Ot(\log(q))$, for any $\varepsilon>0$.  Since
the input sizes are roughly proportional to $\delta_1$
(resp. $\delta_2$) field elements, this means that with respect to
$\delta_1$ (resp. $\delta_2$), we obtain a subquadratic running time
in the algebraic model, and a quasi-linear running time in the boolean
model.

Before discussing further questions, we briefly comment on the
assumption on the characteristic of $\K$. We do need $2,\dots,\delta_1$
(resp. $2,\dots,\delta_2$) to be invertible in $\K$; otherwise, the
algorithm will not work. The stronger requirement that
$2,\dots,\delta_1^2$ (resp. $2,\dots,\delta_2^2$) are units allows us
to find random elements in $\K$ that are ``lucky'' with large
probability; if this assumption does not hold, the algorithm may still
succeed, but we lose the control on the expected running time.

The basic idea of our algorithms is from~\cite{PoSc10}: we reduce
everything to computations with univariate polynomials, since most
operations above will be easy to deal with in the univariate case. To
this end, we perform a change of representation between our input and
a univariate representation, by using repeatedly modular composition
and power projection.

This raises the question of whether better algorithms may be possible,
bypassing modular composition and power projection.  The following
theorem essentially proves that this is not the case, and that
computing the equiprojectable decomposition is essentially equivalent
to modular composition or power projection, at least for the choice of
parameter $m=1$.

In what follows, let $\EE: \N^2 \to \N$ be such that one can solve
problem ${\bf P}_1$ above in $\EE(n,\delta_1)$ base field operations
(in an algebraic model), for triangular sets in $n$ variables. Then,
our second main result is the following.

\begin{Theo}
  \label{theo:EtoC}
  Let $\Tt$ be a triangular set in $n$ variables, with $n\in\{1,2\}$,
  that generates a radical ideal. Then, we can compute modular
  compositions and power projections modulo $\langle \Tt \rangle$ with
  parameters $(1,n)$ and size $\delta_{\bf f} \le \delta_\Tt$ in time
  $2\EE(4,\delta_\Tt)+\Ot(\delta_\Tt)$.
\end{Theo}
\noindent In other words, if we are able to compute four-variate
equiprojectable decompositions efficiently, we can compute modular
compositions and power projections efficiently for some small values
of the parameters (which cover in particular the most useful case
$m=n=1$, that is, computing $F(G) \bmod T$, for univariate polynomials
$F,G,T$). Note that an entirely similar result holds for the boolean
model as well.

\paragraph{Organization of the paper.} 
Section~\ref{sec:notation} introduces most basic algorithms used in
the paper: a reminder on modular composition and power projection for
triangular sets in one or two variables, and conversions between
univariate and triangular representations. Section~\ref{sec:CtoE}
gives an algorithm to compute the so-called $\phi$-decomposition of a
zero-dimensional algebraic set $V$, that is, a decomposition according
to the cardinalities of the fibers of a mapping $\phi:V \to
\Kbar^m$. We use this in Section~\ref{ssec:equi} to prove
Theorem~\ref{theo:CtoE}; in that section, we also present experimental
results obtained with a Maple implementation. Finally,
Section~\ref{sec:EtoC} proves Theorem~\ref{theo:EtoC}.

\paragraph{Previous work.}
Let us first review previous work for the questions we consider in the
algebraic complexity model.

For a triangular set $\Tt$, some previous algorithms have costs of the
form $\Ot(4^n \delta_\Tt)$ for multiplication in
$R_\Tt$~\cite{LiMoSc09} or $\Ot(K^n \delta_\Tt)$ for computing
quasi-inverses in $R_\Tt$~\cite{DaMoScXi06}, for $K$ a large
constant. For multiplication, some particular cases with a better
cost are discussed in~\cite{BoChHoSc09}. An algorithm for
regularization, a similar question to quasi-inverse, is given
in~\cite{LiMoPa09,LiMoPa10}; under a non-degeneracy assumption, its
cost grows like $\sum_{2\le i\le n}2^id_1 \cdots d_{i-1}d_i^{i+1}$, up
to polylogarithmic factors. In particular, all these algorithms
involve an extra factor of the form $K^n$.

For change of order, previous work includes~\cite{BoLeMo01} (which
covers more general questions, e.g. in positive dimension), for which
we are not aware of a complexity analysis. A close reference to our
work is~\cite{PaSc06}: the results in that paper are restricted to the
bivariate case, but use similar techniques; our algorithms are
actually a generalization of those in~\cite{PaSc06}. 

It is worth discussing in some detail a natural approach to change of
order, based on resultant computations. In the simplest case of
bivariate systems, changing the order in a triangular set
$(T_1(X_1),T_2(X_1,X_2))$ can be done by first computing the resultant
${\rm res}(T_1,T_2,X_1)$, so as to eliminate $X_1$ --- this would of
course be only the first step of the algorithm, since we would also
have to deal with $X_2$. Still, already this first step may be costly,
since the best algorithm we are aware of takes time
$O\tilde{~}(d_1^2d_2)$, which can be as large as
$O\tilde{~}(\delta_\Tt^2)$. An extension to triangular sets in more
variables could be done along the lines of~\cite{LiMoPa09,LiMoPa10};
roughly speaking, it may induce costs similar to the one seen above
for regularization.

For the problem of computing the equiprojectable decomposition (or
more generally, for our question ${\bf P}_1$), we are not aware of
previous complexity results.

In the boolean model, relying on the results by Kedlaya and Umans
mentioned above, we showed in~\cite{PoSc10} that it is possible to
answer some of our questions in $n^2
\delta_\Tt^{1+\varepsilon}\Ot(\log(q))$ bit operations, for any fixed
$\varepsilon > 0$ (note that exponential terms of the form $K^n$ have
disappeared).  Those results addressed multiplication in $R_\Tt$ and
some restricted forms of inversion and change of order, but did not
consider any issues related to equiprojectable decomposition.

%%%%%%%%%%%%%%%%%%%%%%%%%%%%%%%%%%%%%%%%%%%%%%%%%%%%%%%%%%%%
%%%%%%%%%%%%%%%%%%%%%%%%%%%%%%%%%%%%%%%%%%%%%%%%%%%%%%%%%%%%
%%%%%%%%%%%%%%%%%%%%%%%%%%%%%%%%%%%%%%%%%%%%%%%%%%%%%%%%%%%%

\section{Notations and known results}\label{sec:notation}

In this section, we first recall a few results from the literature,
and describe algorithms for bivariate modular composition and power
projection (thereby proving the claim made in the introduction
regarding the cost of these operations in an algebraic model). In a
second subsection, we discuss the representation of zero-dimensional
algebraic sets by means of {\em univariate representations}, and give
some basic algorithms for this data structure.

%%%%%%%%%%%%%%%%%%%%%%%%%%%%%%%%%%%%%%%%%%%%%%%%%%%%%%%%%%%

\subsection{Basic algorithms}\label{ssec:basics}

In this subsection, we let $\A$ denote either $\K[X_1]$ or
$\K[X_1,X_2]$ and we consider a triangular set $\Tt$ in $\A$; we write
as usual $R_\Tt=\A/\langle \Tt\rangle$ and we let $V$ be the zero-set
of $\Tt$, in either $\Kbar$ or $\Kbar^2$. We will describe a few
useful algorithms for computing in $R_\Tt$; most of them actually
extend to $\A=\K[X_1,\dots,X_n]$, but the costs would then involve an
extra factor of the form~$K^n$, for some constant $K$.

In all this subsection, we will assume that the characteristic of $\K$
is equal to $0$ or greater than~$\delta_\Tt$.

\paragraph{Multiplication and transposed multiplication.}
Using univariate multiplication, we can do the following in
$O(\M(\delta_\Tt))$ operations in $\K$:
\begin{itemize}
\item {\em modular multiplication:} given $A,B \in R_\Tt$, compute $AB
  \in R_\Tt$
\item {\em transposed multiplication:} given a linear form $\ell:R_\Tt
  \to \K$ and $A\in R_\Tt$, compute the linear form $A\cdot \ell:
  R_\Tt \to \K$ defined by $(A\cdot \ell)(B)=\ell(AB)$.
\end{itemize}
See for instance~\cite{GaSh92} and~\cite{PaSc06} for a proof.

\paragraph{Modular composition.} 
In this paragraph, we discuss modular composition with parameters
$(m,n)$, with $m=2$: given $F \in \K[Y_1,Y_2]$, with $\deg(F,Y_1)<f_1$
and $\deg(F,Y_2)<f_2$, and given $G_1,G_2$ in $R_\Tt$, this amounts to
compute $F(G_1,G_2) \in R_\Tt$. For $(m,n)=(1,1)$, that is, with $F$
univariate and $\Tt=(T_1) \in \K[X_1]$, the best-known algorithm is
due to Brent and Kung~\cite{BrKu78}. We present here a straightforward
generalization, under the simplifying assumption that $f_1 f_2 \le
\delta_\Tt$. Note that solving this problem for $m=2$ actually also
solves it for $m=1$, by taking $f_2=1$.

We let $\varepsilon_1,\varepsilon'_1$ and
$\varepsilon_2,\varepsilon'_2$ be positive integers such that
$\varepsilon_1 \varepsilon'_1 \ge f_1$ and $\varepsilon_2
\varepsilon'_2 \ge f_2$ (to be specified below), and we decompose $F$
into ``rectangular slices'' of the form
\[
F=\sum_{i_1 < \varepsilon_1, i_2 < \varepsilon_2} F_{i_1,i_2}(Y_1,Y_2)
Y_1^{\varepsilon'_1 i_1} Y_2^{\varepsilon'_2 i_2},
\]
with each $F_{i_1,i_2}$ in $\K[Y_1,Y_2]$ and satisfying
$\deg(F_{i_1,i_2},Y_1)<\varepsilon'_1$ and
$\deg(F_{i_1,i_2},Y_2)<\varepsilon'_2$. Then, we have
\[
F(G_1,G_2) =\sum_{i_1 < \varepsilon_1, i_2 <
  \varepsilon_2}\varphi_{i_1,i_2} \gamma_1^{i_1} \gamma_2^{i_2},
\]
with $\varphi_{i_1,i_2}=F_{i_1,i_2}(G_1,G_2)$, $\gamma_1 =
G_1^{\varepsilon'_1}$ and $\gamma_2=G_2^{\varepsilon'_2}$, all
equalities being modulo $\langle \Tt \rangle$. This gives the
following algorithm:
\begin{enumerate}
\item Compute all powers $G_1^{j_1} G_2^{j_2} \bmod \langle
  \Tt\rangle$, for $j_1 < \varepsilon'_1$, $j_2 < \varepsilon'_2$,
  $\gamma_1$, as well as $\gamma_2$. This costs a total of
  $\varepsilon'_1 \varepsilon'_2$ multiplications in $R_\Tt$ (one per
  monomial).
\item We deduce all $\varphi_{i_1,i_2}$ by linear algebra: given
  $(i_1,i_2)$, $\varphi_{i_1,i_2}=F_{i_1,i_2}(G_1,G_2) \bmod \langle
  \Tt\rangle$ is obtained by doing the matrix-vector product $M_G
  V_{i_1,i_2}$, where $M_G$ is the matrix of size $(\delta_\Tt \times
  \varepsilon'_1 \varepsilon'_2)$ that contains the coefficients of
  all $G_1^{j_1} G_2^{j_2} \bmod \langle \Tt \rangle$ (in columns) and
  $V_{i_1,i_2}$ is the column-vector of coefficients of $F_{i_1,i_2}$;
  to do it for all $(i_1,i_2)$, we end up doing one matrix product of
  size $(\delta_\Tt \times \varepsilon'_1 \varepsilon'_2) \times
  (\varepsilon'_1 \varepsilon'_2 \times \varepsilon_1 \varepsilon_2)$.
\item We eventually get $F(G_1,G_2) \bmod \langle \Tt \rangle$ by
  using Horner's scheme twice: first, to compute
  \[
  \varphi_{i_1} = \sum_{i_2 < \varepsilon_2} \varphi_{i_1,i_2}
  \gamma_2^{i_2} \bmod \langle \Tt \rangle, i_1 < \varepsilon_1;
  \]
  this is done with $\varepsilon_2-1$ multiplications modulo $\langle
  \Tt \rangle$. Then to compute
  \[
  F(G_1,G_2) \bmod \langle \Tt \rangle = \sum_{i_1 < \varepsilon_1}
  \varphi_{i_1} \gamma_1^{i_1}.
  \]
  The total is $\varepsilon_1 \varepsilon_2-1$ multiplications modulo
  $\langle \Tt \rangle$.
\end{enumerate}
In total, we do at most $\varepsilon_1 \varepsilon_2+\varepsilon'_1
\varepsilon'_2$ multiplications modulo $\langle \Tt \rangle$ and a
matrix product of size $(\delta_\Tt \times \varepsilon'_1
\varepsilon'_2) \times (\varepsilon'_1 \varepsilon'_2 \times
\varepsilon_1 \varepsilon_2)$.  We take $\varepsilon_1 \simeq
\varepsilon'_1 \simeq {f_1}^{1/2}$ and $\varepsilon_2 \simeq
\varepsilon'_2 \simeq {f_2}^{1/2}$, and we write $\varphi=f_1 f_2$.
Then, we end up with $O(\varphi^{1/2})$ multiplications modulo
$\langle \Tt \rangle$ and a matrix product of size $(\delta_\Tt \times
\varphi^{1/2}) \times (\varphi^{1/2} \times \varphi^{1/2})$.  Since by
assumption $\varphi=O(\delta_\Tt)$, the cost is
$O(\M(\delta_\Tt)\delta_\Tt^{1/2} + \delta_\Tt^{(\omega+1)/2})$, which
is $O(\delta_\Tt^{(\omega+1)/2})$.

\paragraph{Power projection.} 
Next, we present an algorithm to solve the power projection problem
for parameters $(m,n)$, with $m=2$. Recall that power projection takes
as input a linear form $\ell\in R_\Tt^*$, $G_1$ and $G_2$ in $R_\Tt$,
some bounds $(f_1,f_2)$, and outputs the sequence
$(\ell(G_1^{i_1}G_2^{i_2} \bmod \langle \Tt \rangle))_{i_1
  <f_1,i_2<f_2}$.

For parameters $(m,n)=(1,1)$, the algorithm is due to
Shoup~\cite{Shoup99} and an extension to $n=2$ is due to
Kaltofen~\cite{Kaltofen00}; these algorithms are dual to Brent-Kung's
algorithm. As for modular composition, we present a straightforward
generalization to $m=2$, with the assumption $f_1 f_2 \le
\delta_\Tt$. The algorithm is obtained by simply transposing steps 2
and 3 of the modular composition algorithm (step 1 is kept as a
preprocessing phase), so the cost estimate is therefore the same.

Let $\varepsilon_1,\varepsilon'_1,\varepsilon_2,\varepsilon'_2$ be as
above, and let again $\gamma_1=G_1^{\varepsilon'_1} \bmod \langle \Tt
\rangle$ and $\gamma_2=G_2^{\varepsilon'_2} \bmod \langle \Tt
\rangle$. For $i_1 < \varepsilon_1$ and $i_2 < \varepsilon_2$, let
\[\ell_{i_1,i_2} = (\gamma_1^{i_1} \gamma_2^{i_2}) \cdot \ell,\]
where the ``dot'' denotes transposed multiplication. It follows that
for $j_1 < \varepsilon_1'$ and $j_2 < \varepsilon_2'$, we have
\[\begin{array}{rcl}
  \ell_{i_1,i_2} (G_1^{j_1}G_2^{j_2} \bmod \langle \Tt \rangle)
  &=&\ell(\gamma_1^{i_1}\gamma_2^{i_2} G_1^{j_1}G_2^{j_2} \bmod \langle \Tt \rangle)\\
  &=&\ell(G_1^{\varepsilon'_1 i_1 + j_1} G_2^{\varepsilon'_2 i_2 +
    j_2}\bmod \langle \Tt \rangle).
\end{array}\] Thus, we compute all $\ell_{i_1,i_2}
(G_1^{j_1}G_2^{j_2} \bmod \langle \Tt \rangle)$, for $i_1 <
\varepsilon_1$, $i_2 <\varepsilon_2$, $j_1 < \varepsilon'_1$ and $j_2
<\varepsilon'_2$, as this gives us the values we need.
\begin{enumerate}
\item First, we compute all powers $G_1^{j_1} G_2^{j_2} \bmod \langle
  \Tt \rangle$, with $j_1 < \varepsilon'_1$ and $j_2 <
  \varepsilon'_2$. This costs $\varepsilon'_1 \varepsilon'_2-1$
  multiplications modulo $\langle \Tt \rangle$.  We need as well
  $\gamma_1$ and $\gamma_2$, for two extra multiplications.
\item Then, we compute the linear forms $\ell_{i_1,i_2}$ incrementally
  by $\ell_{i_1+1,i_2} = \gamma_1 \cdot \ell_{i_1,i_2}$ and
  $\ell_{i_1,i_2+1} = \gamma_2 \cdot \ell_{i_1,i_2}$; each of them
  takes one transposed multiplication.
\item We finally compute all $\ell_{i_1,i_2} (G_1^{j_1}G_2^{j_2} \bmod
  \langle \Tt \rangle)$ by computing the matrix product $M_L M_G$,
  where $M_G$ is the same $(\delta_\Tt \times \varepsilon'_1
  \varepsilon'_2)$ matrix as in the modular composition case, and
  $M_L$ is the $(\varepsilon_1 \varepsilon_2 \times
  \delta_\Tt)$ matrix giving the coefficients of the
  $\ell_{i_1,i_2}$.
\end{enumerate}

%% %% \footnote{J'ai un doute l\`a : on dit qu'on transpose
%% %%     les \'etapes, mais au final la taille des matrices consid\'er\'ees
%% %%     ne sont pas transpos\'ees l'une de l'autre. Bon vu ce qu'on prend
%% %%     pour $\varepsilon_1,\varepsilon_2,\varepsilon_1'$ et
%% %%     $\varepsilon_2'$ ca le devient \`a la fin, mais y'a rien qui
%% %%     coince dans l'assertion ``on pr\'esente ici la version
%% %%     transpos\'ee'' ?} 
%% %% Non, rien ne coince : le sens direct envoie une matrice
%% %% $\varepsilon'_1\varepsilon'_2 \times \varepsilon_1\varepsilon_2$ sur
%% %% son produit \`a gauche par $M_G$, et le r\'esultat a taille
%% %% $\delta_\Tt \times \varepsilon_1\varepsilon_2$; le sens transpos\'e
%% %% envoie une matrice $\varepsilon_1\varepsilon_2 \times \delta_\Tt$ 
%% %% sur son produit \`a droite par $M_G$, et le r\'esultat a taille
%% %% $\varepsilon_1\varepsilon_2 \times \varepsilon'_1\varepsilon'_2$.

In total, we do $\varepsilon_1 \varepsilon_2+\varepsilon'_1
\varepsilon'_2$ (transposed) multiplications modulo $\langle \Tt
\rangle$ and a matrix product of size $(\varepsilon_1 \varepsilon_2
\times \delta_\Tt) \times (\delta_\Tt \times \varepsilon'_1
\varepsilon'_2)$. Let $\varphi=f_1f_2$. With $\varepsilon_1 \simeq
\varepsilon'_1 \simeq {f_1}^{1/2}$ and $\varepsilon_2 \simeq
\varepsilon'_2 \simeq {f_2}^{1/2}$, we end up with $2\varphi^{1/2}$
(transposed) multiplications modulo $\langle \Tt \rangle$ and a matrix
product of size $ (\varphi^{1/2} \times \delta_\Tt) \times (\delta_\Tt
\times \varphi^{1/2})$. Since $\varphi=O(\delta_\Tt)$, the cost is
$O(\M(\delta_\Tt)\delta_\Tt^{1/2} + \delta_\Tt^{(\omega+1)/2})$, which
is $O(\delta_\Tt^{(\omega+1)/2})$.

Together with the former algorithm for modular composition, this shows
indeed that we can take $\CC(d)$ in $O(d^{(\omega+1)/2})$, as claimed
in the introduction.

\paragraph{Trace and characteristic polynomial.}
For $A \in R_\Tt$, we let $\tau(A)\in \K$ and $\chi_A \in \K[X]$ be
respectively the trace and characteristic polynomial of the
multiplication-by-$A$ endomorphism of $R_\Tt$. We discuss briefly how
to compute these objects.

The trace $\tau:R_\Tt \to \K$ is actually a $\K$-linear form. Using
fast multiplication, it is possible to determine its values on the
monomial basis $B_\Tt$ of $R_\Tt$ using $O(\M(\delta_\Tt))$
operations~\cite{PaSc06}.

Since $R_\Tt$ is a reduced algebra, by~\cite[Prop.~4.2.7]{CoLiSh98}
(sometimes called Stickelberger's Theorem), we have
\begin{equation}\label{eq:stick}
  \chi_A=\prod_{\x \in V}(X-A(\x)).  
\end{equation}
We can compute $\chi_A$ using power projection (this is well-known,
see e.g.~\cite{Rouillier99} for a presentation of this algorithm in a
more general context). We start by computing the values of the trace
$\tau$ on the monomial basis $B_\Tt$. By power projection, we can then
compute the traces $\tau(A^i)$, for $i=0,\dots,\delta_\Tt-1$, which
are the power sums of $\chi_A$. By our assumption on the
characteristic of $\K$, we can then use Newton iteration (for the
exponential of a power series) to deduce the characteristic polynomial
$\chi_A$ of $A$ in time $O(\M(\delta_\Tt))$,
see~\cite{BrKu78,Schonhage82}. By our assumption that $\M(d)\log(d) =
O(\CC(d))$, we deduce that the power projection is the dominant part
of this algorithm, so the total cost is $O(\CC(\delta_\Tt))$.

\paragraph{Inverse modular composition.}
A second use of trace formulas is an inverse modular
composition. Given $A$ and $B$ in $R_\Tt$, we want to compute a
polynomial $U \in \K[X]$, if it exists, such that $B=U(A)$ in
$R_\Tt$. In~\cite{PoSc10}, following ideas
from~\cite{Shoup94,Rouillier99}, we recall an algorithm that computes
a polynomial $U$ in time $O(\CC(\delta_\Tt))$, such that if $B$ can
indeed be written as a polynomial in $A$, then $B=U(A)$; note that the
analysis uses the assumption that $\M(d)\log(d)$ is in $O(\CC(d))$,
and our assumption on the characteristic of $\K$. Verifying whether
$B=U(A)$ can be done for another modular composition, so the total
time is $O(\CC(\delta_\Tt)).$

%%%%%%%%%%%%%%%%%%%%%%%%%%%%%%%%%%%%%%%%%%%%%%%%%%%%%%%%%%%%

\subsection{Univariate representations}\label{ssec:uni}

We next turn to questions related to the representation of
zero-dimensional algebraic sets. We have already introduced triangular
representations; in this subsection, we will discuss {\em univariate
  representations}, which rely on the introduction of a linear
combination of all variables, and for which most of our questions are
easy to solve.

In all that follows, the {\em degree} $\deg(V)$ of a zero-dimensional
algebraic set $V$ simply denotes its cardinality.

\paragraph{Definition.}
Let $V\subset \Kbar^n$ be a zero-dimensional algebraic set of degree
$\delta$, defined over $\K$, and let $I$ be its defining ideal.

A {\em univariate representation} $\Ur=(P,\Uu,\mu)$ of $V$ consists of
a polynomial $P\in\K[X]$, a sequence of polynomials
$\Uu=(U_1,\dots,U_n) \in \K[X]$, with $\deg(U_i) < \deg(P)$ for all
$i$, as well as a linear form $\mu=\mu_1 X_1 + \cdots+ \mu_n X_n$ with
coefficients in $\K$, such that
\begin{equation}\label{eq:psi}
  \begin{array}{cccc}
    \Psi_\Ur:&\K[\X]/I& \to & \K[X]/\langle P\rangle \\
    &X_1,\dots,X_n & \mapsto & U_1,\dots,U_n\\
    &\mu_1 X_1 + \cdots + \mu_n X_n & \mapsfrom & X
  \end{array}
\end{equation}
is an isomorphism: this allows one to transfer most algebraic
operations to the ring $\K[X]/\langle P\rangle$, where arithmetic is
easy. In particular, the definition implies that $P$ is squarefree,
and that it is the characteristic polynomial of $\mu$ in $\K[\X]/I$.
Thus, we have
$$P=\prod_{\x \in V}(X-\mu(\x))$$
and $x_i=U_i(\mu(\x))$ for all $\x=(x_1,\dots,x_n)$ in $V$ and $i \le
n$.

This kind of representation is familiar: up to a few differences, it
is used for instance in
\cite{GiMo89,AlBeRoWo96,Rouillier99,GiHeMoPa95,GiLeSa01}.

We will call a linear form $\mu=\mu_1 X_1 + \cdots+ \mu_n X_n$ a {\em
  separating element} for $V$ if for all distinct $\x,\x'$ in $V$,
$\mu(\x)\ne\mu(\x')$. One easily sees that $\mu$ is separating if and
only if $V$ admits a univariate representation of the form
$\Ur=(P,\Uu,\mu)$, if and only if the characteristic polynomial $P$ of
$\mu$ in $\K[\X]/I$ is squarefree. This characterization implies the
following well-known lemma.
\begin{Lemma}\label{lemma:prob}
  If the characteristic of $\K$ is at least $\delta^2$, and if
  $\mu_1,\dots,\mu_n$ are chosen uniformly at random in ${\frak
    S}=\{0,\dots,\delta^2-1\}$, the probability that $\mu=\mu_1 X_1
  +\cdots +\mu_n X_n$ be a separating element for $V$ is at least
  $1/2$. The same remains true if $\mu_n$ is set to $1$ and
  $\mu_1,\dots,\mu_{n-1}$ are chosen uniformly at random in ${\frak
    S}$.
\end{Lemma}
\begin{proof}
  The above characterization implies that $\mu$ is separating if and
  only if $(\mu_1,\dots,\mu_n)$ does not cancel the polynomial
  $\Delta$ of degree $\delta(\delta-1)/2$ defined by
  $$\Delta(M_1,\dots,M_n)=\prod_{\x,\x' \in V,\ \x\ne\x'} \left
  (M_1(x_1-x'_1) + \cdots + M_n(x_n-x'_n)\right).$$ The
  Zippel-Schwartz lemma implies that there are at most $\delta^{2n}/2$
  roots of $\Delta$ in ${\frak S}^n$, and the first statement
  follows. To get the second one, observe that $\Delta$ is
  homogeneous, so we can set $M_n=1$ without loss of generality; the
  second statement follows.
\end{proof}

\paragraph{Useful algorithms.} We conclude this section with a few
algorithms for univariate representations. Most of what is here is
standard, or at least folklore, although the complexity statements
themselves may be new (e.g., one finds in~\cite{GiLeSa01} an
equivalent of Lemma~\ref{lemma:change} below, but with a quadratic
running time).

\begin{Lemma}\label{lemma:change}
  Given a univariate representation $\Ur=(P,\Uu,\mu)$ of an algebraic
  set $V \subset\Kbar^n$ defined over $\K$, and a linear form
  $\nu=\nu_1 X_1 + \cdots + \nu_n X_n$ with coefficients in $\K$, one
  can decide whether $\nu$ is a separating element for $V$, and if so
  compute the corresponding univariate representation
  $\Vr=(Q,\Vv,\nu)$, in time $O(n \CC(\delta))$, with
  $\delta=\deg(V)$, provided that the characteristic of $\K$ is equal
  to $0$ or greater than $\delta$.
\end{Lemma}
\begin{proof}
  Let $\Psi_{\Ur}$ be as in Equation~\eqref{eq:psi}. We first compute
  $N=\Psi_{\Ur}(\nu)=\nu_1 U_1 + \cdots + \nu_n U_n$; this takes only
  $O(n\delta)$ operations.

  Next, we compute the characteristic polynomial $Q$ of $N$ in
  $\K[X]/\langle P\rangle$; as mentioned before, $\nu$ is a separating
  element for $V$ if and only if $Q$ is squarefree. We have seen that
  computing $Q$ takes time $O(\CC(\delta))$; testing squarefreeness
  takes time $O(\M(\delta)\log(\delta))$, which is by assumption
  $O(\CC(\delta))$.

  When $\mu$ is separating, we can use the algorithm for inverse
  modular composition, to find polynomials $V_1,\dots,V_n$ such that
  $U_i=V_i(N) \bmod Q$ holds for all $i$; then, we have found
  $\Vr=(Q,(V_1,\dots,V_n),\nu)$. In view of the results recalled in
  Subsection~\ref{ssec:basics} on inverse modular composition, the
  total time is $O(n \CC(\delta))$.
\end{proof}

\begin{Lemma}\label{lemma:merge}
  Given univariate representations $\Ur=(P,\Uu,\mu)$ and
  $\Vr=(Q,\Vv,\nu)$ of two algebraic sets $V \subset\Kbar^n$ and $W
  \subset\Kbar^n$ defined over $\K$, one can compute univariate
  representations of either $V \cup W$ or $V-W$ in expected time $O(n
  \CC(\delta))$, with $\delta=\deg(V)+\deg(W)$, provided that the
  characteristic of $\K$ is equal to $0$ or greater than $\delta^2$.
\end{Lemma}
\begin{proof}
  The following process is repeated until success. We pick a random
  linear form $\lambda=\lambda_1 X_1 + \cdots + \lambda_n X_n$ with
  coefficients in ${\frak S}=\{0,\dots,\delta^2-1\}$, and apply the
  algorithm of Lemma~\ref{lemma:change} to $(\Ur,\lambda)$ and
  $(\Vr,\lambda)$.  The cost of this step is $O(n\CC(\delta))$. In
  case of success, we let $\Ur'=(P',\Uu',\lambda)$ and
  $\Vr'=(Q',\Vv',\lambda)$ be the resulting univariate representations
  of $V$ and $W$; if either subroutine fails, we pick another
  $\lambda$.

  At this stage, $\lambda$ is separating for both $V$ and $W$.  Now,
  we compute the polynomial $S=\gcd(P',Q')$, as well as $P''=P'/S$ and
  $Q''=Q'/S$. We also compute
  $$U''_i=U'_i \bmod P'',\quad T_i=U'_i \bmod S,\quad W_i=V'_i \bmod
  S,\quad V''_i=V'_i \bmod Q''$$ for all $i$. Using fast GCD and fast
  Euclidean division, this can be done in time
  $O(\M(\delta)\log(\delta)+n\M(\delta))$, which is negligible
  compared to the cost of the first step.

  These polynomials will allow us to determine whether $\lambda$ is a
  separating element for $V \cup W$. This is the case if and only if
  for any common root $\alpha$ of $P'$ and $Q'$, the equalities
  $U'_i(\alpha)=V'_i(\alpha)$ hold for all $i \le n$, that is, if
  $T_i=W_i$ holds for all $i$. Doing this test takes time
  $O(n\delta)$; if not all equalities hold, we pick another $\lambda$.
  Note that if $\lambda$ is separating for $V \cup W$, it is
  separating for $V- W$.

  Assuming $\lambda$ is a separating element for $V\cup W$, we obtain
  a univariate representation for $V\cup W$ by computing $(P''S
  Q'',(E_1,\dots,E_n),\lambda)$, where $E_i$ is obtained by applying
  the Chinese Remainder Theorem to $(U''_i, T_i, V''_i)$ and moduli
  $(P'',S,Q'')$, for all $i$. Computing these polynomials takes time
  $O(n \M(\delta)\log(\delta))$, which is again
  $O(n\CC(\delta))$. Similarly, we obtain a univariate representation
  for $V - W$ as $(P'',(U''_1,\dots,U''_n),\lambda)$.

  By Lemma~\ref{lemma:prob}, we expect to test $O(1)$ choices of
  $\lambda$ (precisely, at most 2) before finding a suitable one.  As
  a consequence, the expected running time is $O(n\CC(\delta))$.
\end{proof}

To conclude this section, we mention the following result about
conversions between univariate and triangular representations.

As a preliminary, remember that if $\Ur=(P,\Uu,\mu)$ is a univariate
representation of an algebraic set $V$, there exists an isomorphism
$\Psi_\Ur: \K[\X]/I(V) \to \K[X]/\langle P \rangle$. If furthermore
the defining ideal of $V$ admits a triangular set of generators $\Tt$
for some variable order $<$, we also have $\K[\X]/I(V) \simeq
R_\Tt$. As a result, there exists  change-of-basis isomorphisms
$$\Phi_{\Tt,\Ur}: \K[X]/\langle P \rangle\to R_\Tt 
\quad\text{and}\quad \Psi_{\Tt,\Ur}:R_\Tt \to \K[X]/\langle P
\rangle,$$ which will be useful in the sequel.

\begin{Lemma}\label{lemma:conv}
  Let $V \subset\Kbar^n$ be an algebraic set of degree $\delta$,
  defined over $\K$, and let $I \subset \K[X_1,\dots,X_n]$ be its
  defining ideal; suppose that the characteristic of $\K$ is equal to
  $0$ or greater than $\delta^2$. Let finally $<$ be an order on the
  variables $X_1,\dots,X_n$ and suppose that $I$ is generated by a
  triangular set $\Tt$ for the variable order $<$. Then the following
  holds:
  \begin{itemize}
  \item Given a univariate representation $\Ur=(P,\Uu,\mu)$ of $V$,
    one can compute the triangular set $\Tt$ in expected time
    $O(n^2\CC(\delta))$. Given $A$ in $\K[X]/\langle P\rangle$, one
    can then compute $\Phi_{\Tt,\Ur}(A) \in R_\Tt$ in time
    $O(n\CC(\delta))$.

  \item Given $\Tt$, one can compute a univariate representation
    $\Ur=(P,\Uu,\mu)$ of $V$ in expected time $O(n^2\CC(\delta))$.
    Given $A$ in $R_\Tt$, one can then compute $\Psi_{\Tt,\Ur}(A) \in
    \K[X]/\langle P\rangle$ in time $O(n\CC(\delta))$.
  \end{itemize}
\end{Lemma}
\begin{proof}
  We will merely describe the main ideas, so as to highlight the roles
  of modular composition and power projection. Details are given in
  \cite[Section~5.3 and~6.3]{PoSc10}, together with worked-out
  examples (the complexity analysis there is given in the boolean
  model, but carries over to the algebraic model without
  difficulty). In both directions, we proceed one variable at a time.
  \begin{itemize}
  \item In the first direction, we change (if needed) the linear form
    $\mu$, so as to ensure that the coefficient of $X_n$ in $\mu$ is
    equal to~$1$; this is done in expected time $O(n\CC(\delta))$ by
    means of Lemmas~\ref{lemma:prob} and~\ref{lemma:change}. This mild
    condition is needed to apply the algorithm of~\cite{PoSc10}; we
    still write the input $\Ur=(P,\Uu,\mu)$.

    Then, we let $\mu'=\mu'_1 X_1+ \cdots+\mu'_{n-2}X_{n-2}+X_{n-1}$
    be a random combination of $X_1,\dots,X_{n-1}$, with coefficients
    in $\{0,\dots,\delta^2-1\}$, whose coefficient in $X_{n-1}$ is
    $1$. We can then replace the single polynomial $P_{n}(X)=P(X)$ by
    a bivariate triangular set
    \[ \left | \begin{array}{l}
        T_{n-1,n}(X,X_n) \\
        P_{n-1}(X),
      \end{array}\right .\]
    where $P_{n-1}$ is the squarefree part of the characteristic
    polynomial of $\mu'_1 U_1+ \cdots+\mu'_{n-2}U_{n-2}+U_{n-1}$
    modulo $P_n$. As we go, we also compute expressions of
    $U_1,\dots,U_{n-1}$ as polynomials in $\mu'$, to allow the process to 
    continue. In the second step, we introduce a triangular set
    \[ \left | \begin{array}{l}
        T_{n-2,n}(X,X_{n-1},X_n) \\
        T_{n-2,n-1}(X,X_{n-1})\\
        P_{n-2}(X)
      \end{array}\right .\]
    in three variables $X,X_{n-1},X_n$, and so on until we
    obtain~$\Tt$.

    Using formulas from~\cite{PaSc06,PoSc10}, going from $(P_{n})$ to
    $(P_{n-1},T_{n-1,n})$ is done by means of power projections with
    parameters $(1,1)$ and $(2,1)$ and size $\delta =\deg(P_{n})$, as
    well as inverse modular compositions, all computed modulo $\langle
    P_{n}\rangle$; the total time is $O(n \CC(\delta))$. The change of
    basis $\K[X]/\langle P_{n}\rangle \to \K[X,X_n]/\langle
    P_{n-1},T_{n-1,n}\rangle$ is done by means of a modular
    composition with parameters $(1,2)$ and size $\delta
    =\deg(P_{n})$, computed modulo $\langle P_{n-1},T_{n-1,n}
    \rangle$; it takes time $O(\CC(\delta))$.

    The further steps are done in the same manner. For instance, going
    from $(P_{n-1},T_{n-1,n})$ to $(P_{n-2},T_{n-2,n-1},T_{n-2,n})$
    requires first to compute $(P_{n-2},T_{n-2,n-1})$, similarly to
    what we did in the first step. Then, we obtain $T_{n-2,n}$ by
    applying the change of basis $\K[X]/\langle P_{n-1}\rangle \to
    \K[X,X_n]/\langle P_{n-2},T_{n-2,n-1}\rangle$ to all coefficients
    of $T_{n-1,n}$.

    There are $n$ such steps before we reach $\Tt$; each takes an
    expected $O(n\CC(\delta))$, so the total time is an expected
    $O(n^2\CC(\delta))$.

    Staring from $A$ in $\K[X]/\langle P\rangle$, we obtain its image
    in $R_\Tt$ by computing its representations in $\K[X,X_n]/\langle
    P_{n-1},T_{n-1,n}\rangle$, and so on. Each conversion is done as
    above by means of modular compositions with parameters $(1,2)$ and
    takes time $O(\CC(\delta))$; the total number of operations is thus
    $O(n\CC(\delta))$.

  \item To compute a univariate representation starting from a
    triangular set $\Tt=(T_1,\dots,T_n)$, we follow the same process
    backward. Starting from $(T_{n,1},\dots,T_{n,n})=(T_1,\dots,T_n)$,
    we first work with $(T_{n,1},T_{n,2})$, and find a univariate
    representation for these two polynomials; this gives us the
    triangular set in $n-1$ variables
    $(P_{n-1},T_{n-1,3},\dots,T_{n-1,n})$. We continue until we reach
    a single polynomial $P_n$, which we will simply write $P$.

    The polynomial $P_{n-1}(X)$ is the characteristic polynomial of a
    random combination of $X_1,X_2$ with coefficients in
    $\{0,\dots,\delta^2-1\}$, computed modulo $\langle
    T_{n,1},T_{n,2}\rangle$; all other polynomials $T_{n-1,j}$ are
    obtained by applying the change-of-basis $\K[X_1,X_2]/\langle
    T_{n,1},T_{n,2}\rangle \to \K[X]/\langle P_{n-1}\rangle$.

    This first step requires a power projection with parameters
    $(1,2)$, as well as modular compositions with parameters $(2,1)$,
    and the cost is an expected $O(n\CC(\delta))$. Since there are $n$
    such steps, the total cost is then an expected
    $O(n^2\CC(\delta))$. The change-of-basis $R_\Tt \to \K[X]/\langle
    P \rangle$ is obtained similarly by means of modular compositions,
    and takes time $O(n\CC(\delta))$.
  \end{itemize}
\end{proof}

%%%%%%%%%%%%%%%%%%%%%%%%%%%%%%%%%%%%%%%%%%%%%%%%%%%%%%%%%%%%
%%%%%%%%%%%%%%%%%%%%%%%%%%%%%%%%%%%%%%%%%%%%%%%%%%%%%%%%%%%%
%%%%%%%%%%%%%%%%%%%%%%%%%%%%%%%%%%%%%%%%%%%%%%%%%%%%%%%%%%%%

\section{The $\phi$-decomposition}\label{sec:CtoE}

In this section, we define the notions of {\em $\phi$-equiprojectable}
sets and {\em $\phi$-decomposition} of a zero-dimensional algebraic
set $V\subset \Kbar^n$, where $\phi$ is a mapping $\Kbar^n \to
\Kbar^m$. We then give an algorithm to compute the
$\phi$-decomposition of $V$, by reducing again this problem to
(mainly) modular composition and power projection.

In what follows, we suppose that $V$ is a zero-dimensional algebraic
subset of $\Kbar^n$ of cardinality $\delta$, defined over $\K$, and we
let $I \subset \K[\X]=\K[X_1,\dots,X_n]$ be its defining ideal. We
make the assumption that the characteristic of $\K$ is equal to $0$ or
greater than $\delta^2$.

We start with the definition of some counting functions. Let $\phi$ be
a mapping $\Kbar^n\to \Kbar^m$, given by polynomials with coefficients
in $\K$. For $\x$ in $V$, we let $c(V,\x, \phi)$ be the cardinality of
the set $\{\x' \in V, \phi(\x')=\phi(\x)\}$: this is the number of
points $\x'$ in $V$ such that $\phi(\x')=\phi(\x)$. Then, we say that
$V$ is {\em $\phi$-equiprojectable} if there exists a positive integer
$d$ such that for all $\x$ in $V$, $c(V,\x,\phi)=d$.

In general, we should not expect $V$ to be $\phi$-equiprojectable.
Then, we define
\[
\mathscr{C}(V,\phi,r)=\{\x \in V,\ c(V,\x,\phi)=r\};
\]
this is the set of all $\x \in V$ with $r$ points in their
$\phi$-fiber. Since $V$ is finite, $\x\mapsto c(V,\x,\phi)$ takes only
finitely many values on $V$, say $r_1 < \cdots < r_s$.  As a
consequence, the sets
\begin{equation}\label{eq:Vri}
  V_{r_1} = \mathscr{C}(V,\phi,r_1),\quad \dots,\quad V_{r_s} = \mathscr{C}(V,\phi,r_s) 
\end{equation}
form a partition of $V$; by construction, all these sets are
$\phi$-equiprojectable. We will write
$$\Dec(V,\phi) =\{V_{r_1},\dots,V_{r_s}\},$$ and we will call this
decomposition the {\em $\phi$-decomposition} of $V$.  Although it may
not be clear from our definition, all $V_{r_i}$ are in fact defined
over $\K$.

\begin{Lemma}\label{lemma:def}
  With notation as in~\eqref{eq:Vri}, $V_{r_1},\dots,V_{r_s}$ are
  defined over $\K$.
\end{Lemma}
\begin{proof}
  We are going to prove that for any $r \ge 1$,
  $$\mathscr{C}'(V,\phi,r)=\{\x \in V,\ c(V,\x,\phi)\ge r\}$$ is defined over
  $\K$.  Since $\mathscr{C}(V,\phi,r)=\mathscr{C}'(V,\phi,r)-\mathscr{C}'(V,\phi,r+1)$, and since
  the set-theoretic difference of two zero-dimensional algebraic sets
  defined over $\K$ is still defined over $\K$, this will be
  sufficient to establish our claim.

  Fix $r \ge 1$, and let $V^{(r)}$ be the $r$-fold product $V \times
  \cdots \times V \subset \Kbar^{nr}$; obviously, $V^{(r)}$ is defined
  over $\K$. Let $(\x_1,\dots,\x_r)$ be the coordinates on
  $\Kbar^{nr}$, where each $\x_i$ has length $n$, and let
  $$W^{(r)}=V^{(r)}-\cup_{1 \le i < j \le r} \Delta_{i,j},$$ where
  $\Delta_{i,j}$ is defined by $\x_i=\x_j$. Again, $W^{(r)}$ is
  defined over $\K$, and $(\x_1,\dots,\x_r)$ is in $W^{(r)}$ if and
  only if all $\x_i$ are in $V$ and pairwise distinct. Finally, we
  define
  $$Z^{(r)}=W^{(r)} \cap_{1 \le i < j \le r} V(\phi(\x_i)-\phi(\x_j));$$
  then, $\mathscr{C}'(V,\phi,r)$ is the projection of $Z^{(r)}$ on the first
  factor $\Kbar^n$, so it is indeed defined over $\K$, as claimed.
\end{proof}

Before discussing an algorithm that computes $\Dec(V,\phi)$, we prove
a simple lemma that will be used in the next section.
\begin{Lemma}\label{lemma:comp}
  Consider two mappings $\phi:\Kbar^n \to \Kbar^m$ and $\psi:\Kbar^n
  \to \Kbar^p$, such that $\psi = f \circ \phi$, for some mapping $f:
  \Kbar^m \to \Kbar^p$, and suppose that $V$ is
  $\phi$-equiprojectable. Then any $V'$ in $\Dec(V,\psi)$ is both
  $\phi$-equiprojectable and $\psi$-equiprojectable.
\end{Lemma}
\begin{proof} Let $d$ be the common cardinality of the fibers of the
  restriction of $\phi$ to $V$. Let further $V'$ be in $\Dec(V,\psi)$,
  and let $\x$ be in $V'$. We will show that $c(V',\x,\phi)=d$,
  thereby establishing that $V'$ is $\phi$-equiprojectable ($V'$ is
  $\psi$-equiprojectable by construction).
  
  Remember that $c(V',\x,\phi)$ is the cardinality of the fiber
  $F'=\{\x' \in V', \phi(\x')=\phi(\x)\}$. We claim that we actually
  have $F'=F$, with $F=\{\x' \in V, \phi(\x')=\phi(\x)\}$. Since by
  assumption $|F|=d$, proving $F=F'$ is sufficient to prove that
  $c(V',\x,\phi)=d$.

  Of course, $F'$ is a subset of $F$. Conversely, let $\x'$ be in $F$.
  Then, $\phi(\x)=\phi(\x')$ and our assumption on $\phi$ and $\psi$
  implies that $\psi(\x)=\psi(\x')$. This implies that $\x'$ is in
  $V'$, as claimed.
\end{proof}

We now explain how to compute $\Dec(V,\phi)$.  For simplicity, we will
assume that $m \le n$, and that $\phi$ is a simple linear map (the
algorithm would not be substantially different in general, but a few
extra terms could appear in the cost analysis).
\begin{Prop}\label{prop:dec}
  Consider an algebraic set $V\subset \Kbar^n$ defined over $\K$ and
  of degree $\delta$, and a univariate representation
  $\Ur=(P,\Uu,\mu)$ of $V$, and let
  $\Dec(V,\phi)=\{V_{r_1},\dots,V_{r_s}\}$. Suppose that the following
  conditions are satisfied:
  \begin{itemize}
  \item the characteristic of $\K$ is equal to $0$ or greater than
    $\delta^2$,
  \item $\phi$ is a linear map $\Kbar^n \to \Kbar^m$, of the form
    $\phi(x_1,\dots,x_n)=(x_1,\dots,x_m)$.
  \end{itemize}
  Then we can compute univariate representations $(P_k,\Uu_k,\mu)_{1
    \le k \le s}$ of $V_{r_1},\dots,V_{r_s}$ in expected time
  $O(\CC(\delta)(n+\log(\delta)) )$.
\end{Prop}
The rest of this section is devoted to prove this proposition. In what
follows, we write $W=\phi(V)$ and, for all $k \le s$,
$W_{r_k}=\phi(V_{r_k})$. We also write $\Uu=(U_1,\dots,U_n)$, with all
$U_i$ in $\K[X]$. Since for all $\x=(x_1,\dots,x_n)$ in $V$ we have
$x_i=U_i(\mu(\x))$, we deduce that
$\phi(\x)=(U_1(\mu(\x)),\dots,U_m(\mu(\x)))$ for $\x$ in $V$.

\paragraph{Step 1.}
Choose a random linear form $\nu= \nu_1 Y_1 + \cdots + \nu_m Y_m$ with
coefficients in $\{0,\dots,\delta^2-1\}$, compute $N=\nu_1 U_1 +
\cdots + \nu_m U_m$, and compute the characteristic polynomial
$\chi_{N}$ of $N$ in $\K[X]/\langle P\rangle$. Computing $N$ takes
time $O(n\delta)$ and computing its characteristic polynomial takes
time $O(\CC(\delta))$, see Subsection~\ref{ssec:basics}.

The linear form $\nu$ must be a separating element for $W$. To verify
if this is the case, we check whether $U_1,\dots,U_m$ can be written
as polynomials in $N$ modulo $P$. This is done using the algorithm for
inverse modular composition, and takes time $O(m \CC(\delta))$, which
is $O(n \CC(\delta))$. Due to our assumption on the characteristic of
$\K$, we need to test an expected $O(1)$ choices of $\nu$ before
finding a separating element, see Lemma~\ref{lemma:prob}.

Remark that for $\x$ in $V$, $N(\mu(\x))=\nu_1 U_1(\mu(\x)) + \cdots +
\nu_m U_m(\mu(\x))=\nu(\phi(\x))$.

\paragraph{Step 2.} Compute the squarefree decomposition of
$\chi_N$; this takes times $O(\M(\delta)\log(\delta))$,
see~\cite[Chapter~14]{GaGe03}.  Using the previous notation, we claim
this decomposition has the form
\[
\chi_N = C_1^{r_1} \cdots C_s^{r_s}, \quad\text{with}\quad C_k =
\prod_{\y \in W_{r_k}} (X-\nu(\y)).
\]
Indeed, by Stickelberger's Theorem, we have the factorization
\[
\begin{array}{rcl}
  \chi_N &=& \prod_{\x \in V} (X-N(\mu(\x)))\\[2mm]
  &=& \prod_{\x \in V} (X-\nu(\phi(\x))).
\end{array}
\]
For $\y\in W$, let $r(\y)$ be the cardinality of the fiber
$\phi^{-1}(\y) \cap V$. Then we obtain the factorization
\[
\begin{array}{rcl}
  \chi_N&=& \prod_{\y \in W} (X-\nu(\y))^{r(\y)} \\[2mm]
  &=& \prod_{k \le s} \prod_{\y \in W_{r_k}} (X-\nu(\y))^{r_k},
\end{array}
\]
since by construction the projections $W_{r_k}$ are pairwise
disjoint. As $\nu$ is separating for $W$, the linear factors
$X-\nu(\y)$ are pairwise distinct, which proves our claim.

For future use, note that $\sum_{i \le s} \deg(C_i) \le \delta$, since
$\chi_N=C_1^{r_1} \cdots C_s^{r_s}$ has degree $\delta$.

\paragraph{Step 3.} For $k \le s$, compute $P_k=\gcd(C_k(N),P)$.  We
will prove at the end of the section that this can be done in time
$O(\CC(\delta)\log(\delta))$. That proof will be somewhat lengthy; for
the moment, we will only prove that for $k\le s$, we have
\begin{equation}\label{eq:pk}
  P_k = \prod_{\x \in V_{r_k}} (X-\mu(\x)).
\end{equation}
Both sides are squarefree (since they divide $P$), so to prove our
claim it is enough to prove that the roots of $P_k$ are exactly the
values $\mu(\x)$ for $\x$ in $V_{r_k}$. As a preliminary remark,
recall that for all $\x$ in $V$, we have $\nu(\phi(\x))=N(\mu(x))$.
\begin{itemize}
\item For $\x=(x_1,\dots,x_n)$ in $V_{r_k}$, $\phi(\x)$ is in
  $W_{r_k}$ so $\nu(\phi(\x))$ is a root of $C_k$. By the remark
  above, this shows that $\mu(\x)$ is a root of $C_k(N)$. But of
  course $\mu(\x)$ is also a root of $P$, so $\mu(\x)$ is a root of
  $P_k$.
\item Conversely, consider a root $\alpha$ of $P_k$. Since any root of
  $P_k$ is a root of $P$, $\alpha$ is of the form $\mu(\x)$ for some
  $\x$ in $V$. But by assumption $\alpha=\mu(\x)$ is also a root of
  $C_k(N)$, which means that $\nu(\phi(\x))$ is a root of $C_k$.  In
  particular, $\nu(\phi(\x))$ is a root of no other $C_{k'}$, because
  these polynomials are pairwise coprime. This implies that $\phi(\x)$
  belongs to no other $W_{r_{k'}}$, so it must belong to $W_{r_k}$; thus,
  $\x$ is in $V_{r_k}$.
\end{itemize}
Note also that we have $P=P_1 \cdots P_s$, all $P_k$ being pairwise coprime.

\paragraph{Step 4.} For $k \le s$ and $j \le n$, compute $U_{k,j} =
U_j \bmod P_k$. This can be done in time $O(n \M(\delta)
\log(\delta))$ using fast multiple
reduction~\cite[Chapter~10]{GaGe03}, which is $O(n\CC(\delta))$.
Writing $\Uu_k=(U_{k,1},\dots,U_{k,n})$, Eq.~\eqref{eq:pk} shows that
for $k \le s$, $(P_k,\Uu_k,\mu)$ is a univariate representation of
$V_{r_k}$, so we are done.

\paragraph{Analysis of Step 3.} Summing all the costs mentioned above
gives the cost estimate claimed in Proposition~\ref{prop:dec}.  All
that is missing is to prove that, as announced, the cost of computing
the polynomials $P_k$ of Step~3 is $O(\CC(\delta)\log(\delta))$.

Recall that for all $k \le s$, $P_k=\gcd(C_k(N),P)$. We cannot compute
the polynomials $C_k(N)$, or even $C_k(N)\bmod P$, as there are too
many of them: one easily sees that $s$ could be as large as
$\sqrt{\delta}$; each polynomial $C_k(N)\bmod P$ requires to store
$\delta$ field elements, so computing all of them would take time at
least $\delta^{1.5}$.

Therefore, we compute the $P_k$ directly, using divide-and-conquer
techniques. Given polynomials $A,Q \in \K[X]$, we will write
\begin{eqnarray}
  \label{eq:gamma1}\Gamma(A,Q)&=&\gcd(A(N),\,Q)\\
  \label{eq:gamma2}&=&\gcd(A(N \bmod Q)\bmod Q,\,Q),
\end{eqnarray} so that the polynomials we want to compute are
$P_1=\Gamma(C_1,P),\dots,P_s=\Gamma(C_s,P)$. 

Assuming we know $N \bmod Q$, Definition~\eqref{eq:gamma2} shows that
we can compute $\Gamma(A, Q)$ by computing first $A(N \bmod Q)\bmod
Q$, then taking its GCD with $Q$. Since by assumption $\M(d)\log(d)$
is $O(\CC(d))$, we can thus obtain $\Gamma(A, Q)$ in time $O(\CC(d))$
by modular composition and fast GCD, with $d=\max(\deg(A),\deg(Q))$;
we will call this the {\em plain algorithm}. In particular, we could
compute any $P_k$ in time $O(\CC(\delta))$. However, as we mentioned
above, computing all $P_k$ directly in this manner incurs a cost of
the form $s \CC(\delta)$, which is too much for our purposes.

The key equality we will use is the following: for any polynomials
$A,B$, we have
\begin{equation}\label{eq:gamma}\Gamma(A, Q) = \Gamma(A, \Gamma(AB,
  Q)).\end{equation}
Indeed, using Definition~\eqref{eq:gamma1}, the left-hand side
reads
$$\Gamma(A, Q)=\gcd(A(N),\,Q),$$
whereas the right-hand side is
$$\Gamma(A, \Gamma(AB, Q))=\gcd(A(N),\,\gcd((AB)(N),\,Q));$$
equality~\eqref{eq:gamma} follows from the fact that
$\gcd(F_1,G)=\gcd(F_1,\gcd(F_1F_2,G))$ holds for all polynomials
$F_1,F_2,G$.

We are now ready to explain how to complete Step~3.  To simplify our
presentation, we will assume that $s$ is a power of two, of the form
$s=2^w$; when this is not the case, we can complete $C_1,\dots,C_s$ by
dummy polynomials $C_k=1$, so as to replace $s$ by the next power of
two, without affecting the asymptotic running time.

\bigskip\noindent{\em Step 3.1.~} We compute the {\em subproduct tree}
(see details below) associated to
$C_1,\dots,C_s$. From~\cite[Chapter~10]{GaGe03}, this can be done in
time $O(\M(\delta)\log(\delta))$, since we have seen that $\sum_{i \le
  s} \deg(C_i) \le \delta$. Using our assumption on $\M$ and $\CC$,
this is in $O(\CC(\delta))$.

At the top level of the subproduct tree, the root is labelled by
$K_{0,1}=C_1\cdots C_s$; its two children are labelled by $K_{1,1}=C_1
\cdots C_v$ and $K_{1,2}=C_{v+1} \cdots C_s$ with $v=s/2$, and so
on. For $j=0,\dots,w$, the polynomials the $j$th level are written
$K_{j,i}$, with $i=1,\dots,2^j$, so that $K_{j,i}=K_{j+1,2i-1}
K_{j+1,2i}$. At the leaves, for $j=w$, we have $K_{w,i}=C_{i}$.

In what follows, we are going to compute all polynomials
$\Gamma(K_{j,i},P)$, for $j=0,\dots,w$ and $i=1,\dots,2^j$, in a
top-down manner. At the leaves, for $j=w$, we will obtain the
polynomials $\Gamma(K_{w,i},P) =\Gamma(C_i,P)=P_i$ we are looking for.

\bigskip\noindent{\em Step 3.2.~} We compute
$$\gamma_{0,1} = \Gamma(K_{0,1},P)$$ using the plain algorithm, in
time $O(\CC(\delta))$, as well as $N_{0,1} = N \bmod \gamma_{0,1}$ in
time $O(\M(\delta))$, by fast Euclidean division. The latter cost is
negligible.

\bigskip\noindent{\em Step 3.3.~} For $j=0,\dots,w-1$ and
$i=1,\dots,2^j$, assuming we know $\gamma_{j,i}$ and $N_{j,i}=N \bmod
\gamma_{j,i}$, we compute
$$\gamma_{j+1,2i-1} = \Gamma(K_{j+1,2i-1}, \gamma_{j,i})
\quad\text{and}\quad \gamma_{j+1,2i} =
\Gamma(K_{j+1,2i},\gamma_{j,i})$$ followed by
$$N_{j+1,2i-1} = N_{j,i} \bmod \gamma_{j+1,2i-1} \quad\text{and}\quad
N_{j+1,2i} = N_{j,i} \bmod \gamma_{j+1,2i}.$$ Our claim is twofold:
first, we will prove that $\gamma_{j,i}=\Gamma(K_{j,i},P)$ for all
$j,i$; second, we will establish that the total running time is
$O(\CC(\delta)\log(\delta))$. Note that this is enough to finish the
proof of Proposition~\ref{prop:dec}, since we have seen that for
$j=w$, we have $\Gamma(K_{w,i},P)=P_i$.

The proof that $\gamma_{j,i}=\Gamma(K_{j,i},P$) is done by induction
on $j$. By definition, this is true for $\gamma_{0,1}$; for $j > 1$,
this follows from Equation~\eqref{eq:gamma}, first taking
$A=K_{j+1,2i-1}$, $B=K_{j+1,2i}$ and $Q=P$, then $A=K_{j+1,2i}$,
$B=K_{j+1,2i-1}$ and $Q=P$.  Since $\gamma_{j+1,2i-1}$ and
$\gamma_{j+1,2i}$ divide $\gamma_{j,i}$, we can also prove by
induction that $N_{j,i}=N \bmod \gamma_{j,i}$ holds for all $j,i$.

It remains to do the cost analysis. Since $N_{j,i}=N \bmod
\gamma_{j,i}$ is known, we can indeed compute $\gamma_{j+1,2i-1}$ and
$\gamma_{j+1,2i}$ from $K_{j+1,2i-1},\,K_{j+1,2i}$ and $\gamma_{j,i}$
by the plain algorithm in time $O(\CC(d_{j,i}))$, where we write
$$d_{j,i}=\max(\deg(K_{j+1,2i-1}),\,
\deg(K_{j+1,2i}),\,\deg(\gamma_{j,i})) \le \max
(\deg(K_{j,i}),\,\deg(\gamma_{j,i})).$$ The computation of
$N_{j+1,2i-1}$ and $N_{j+1,2i}$ can be done in time
$O(\M(\deg(\gamma_{j,i})))$, which is negligible by assumption.
Hence, the total cost is, up to a constant factor,
$$\sum_{j=0,\dots,w-1} \sum_{i=1,\dots,2^j} 
\CC(\max(\deg(K_{j,i}),\,\deg(\gamma_{j,i}))).$$ This admits the
obvious upper bound
$$\sum_{j=0,\dots,w-1} \sum_{i=1,\dots,2^j} \CC(\deg(K_{j,i})) +
\sum_{j=0,\dots,w-1} \sum_{i=1,\dots,2^j} \CC(\deg(\gamma_{j,i})).$$
Using the super-linearity of $\CC$, we obtain the upper bound
$$\sum_{j=0,\dots,w-1}\CC\left( \sum_{i=1,\dots,2^j} 
  \deg(K_{j,i}) \right ) + \sum_{j=0,\dots,w-1}\CC\left(
  \sum_{i=1,\dots,2^j} \deg(\gamma_{j,i})\right).$$ To conclude the
cost analysis, we will prove the inequalities
$$\sum_{i \le 2^j} \deg(K_{j,i}) \le \delta
\quad\text{and}\quad\sum_{i \le 2^j} \deg(\gamma_{j,i})\le \delta.$$
These inequalities imply a cost upper bound of the form
$\sum_{j=0,\dots,w-1}\CC(\delta)$, up to a constant factor. The claim
on the total cost follows, since $w$ is in $O(\log(\delta))$.
\begin{itemize}
\item The first inequality $\sum_{i \le 2^j} \deg(K_{j,i})\le \delta$
  is a straightforward consequence of the equality $\sum_{i \le 2^j}
  \deg(K_{j,i}) = \sum_{i \le s} \deg(C_i)$, which itself follows from
  the definition of the subproduct tree, and the fact that $\sum_{i
    \le s} \deg(C_i) \le \delta$.
\item To obtain the second inequality $\sum_{i \le 2^j}
  \deg(\gamma_{j,i})\le \delta$, we start by proving that for fixed
  $j$, and for $i \ne i'$, $\gamma_{j,i}$ and $\gamma_{j,i'}$ are
  coprime. Indeed, we have seen that
  $$\gamma_{j,i}=\gcd(K_{j,i}(N), P),$$ where $K_{j,i}$ has the form
  $K_{j,i}=\prod_{\ell \in \kappa_{j,i}} C_\ell$. Here, $\kappa_{j,i}$
  is a set of indices which we will not need to make explicit;
  however, for further use, we note that for $i \ne i'$,
  $\kappa_{j,i}$ and $\kappa_{j,i'}$ are disjoint. The factorization
  of $K_{j,i}$ implies that
  $$\gamma_{j,i}=\gcd(\prod_{\ell \in \kappa_{j,i}} C_\ell(N), P).$$
  Recall now that the polynomials $P_\ell=\gcd(C_\ell(N),P)$ are
  pairwise coprime; as a result, the former equality gives
  $$\gamma_{j,i}=\prod_{\ell \in \kappa_{j,i}} \gcd(C_\ell(N), P)=
  \prod_{\ell \in \kappa_{j,i}} P_\ell.$$ Since for fixed $j$ the sets
  $\kappa_{j,i}$ are pairwise disjoint, and since the polynomials
  $P_\ell$ are pairwise coprime, we deduce that for fixed $j$, the
  polynomials $\gamma_{j,i}$ themselves are pairwise coprime, as
  claimed.

  Since by construction all $\gamma_{j,i}$ divide $P$, the product
  $\prod_{i \le 2^j} \gamma_{j,i}$ must divide $P$ as well, and the
  inequality $\sum_{i \le 2^j} \deg(\gamma_{j,i}) \le \delta$ follows.
\end{itemize}

%%%%%%%%%%%%%%%%%%%%%%%%%%%%%%%%%%%%%%%%%%%%%%%%%%%%%%%%%%%%
%%%%%%%%%%%%%%%%%%%%%%%%%%%%%%%%%%%%%%%%%%%%%%%%%%%%%%%%%%%%
%%%%%%%%%%%%%%%%%%%%%%%%%%%%%%%%%%%%%%%%%%%%%%%%%%%%%%%%%%%%

\section{Proof of Theorem~\ref{theo:CtoE}}\label{ssec:equi}

In this section, we prove Theorem~\ref{theo:CtoE}.  We start by
defining equiprojectable sets and the equiprojectable
decomposition. The algorithms underlying Theorem~\ref{theo:CtoE} are
then straightforward applications of the results of the previous
section.

%%%%%%%%%%%%%%%%%%%%%%%%%%%%%%%%%%%%%%%%%%%%%%%%%%%%%%%%%%%%

\subsection{The equiprojectable decomposition}

Let $V\subset \Kbar^n$ be a zero-dimensional algebraic set defined
over $\K$. We suppose that we are given an order $<$ on the variables;
up to renaming them, we can suppose that the order is simply $X_1 <
\cdots < X_n$. For $1 \le i \le n$, we define the projection
\[
\begin{array}{lccc}
  \pi_{i}:& \Kbar^n & \to &\Kbar^i \\ & \x=(x_1,\dots,x_n) & \mapsto &
  (x_1,\dots,x_i).
\end{array}
\]
Then, we say that $V$ is equiprojectable if it is
$\pi_i$-equiprojectable for $i=1,\dots,n$; in other words, $V$ is
equiprojectable if all fibers of $\pi_1$ on $V$ have a common
cardinality $\delta_1$, all fibers of $\pi_2$ on $V$ have a common
cardinality $\delta_2$, etc.

In general, we should not expect $V$ to be equiprojectable. There are
potentially many ways to decompose $V$ into equiprojectable sets; the
equiprojectable decomposition will be a canonical partition of $V$
into pairwise disjoint equiprojectable sets, that will all be defined
over $\K$.

We will actually define a sequence $\Dec(V,i,<)$, for $i=n,\dots,1$,
which will all be partitions of $V$, refining one another. At index
$n$, we write $\Dec(V,n,<) = \{V\}$. Then, for $i<n$, assuming that we
have defined
$$\Dec(V,i+1,<)=\{V_{i+1,1},\dots,V_{i+1,s_{i+1}}\},$$
we obtain $\Dec(V,i,<)$ by computing the $\pi_i$-decomposition of
every element in $\Dec(V,i+1,<)$:
$$\Dec(V,i,<) =\cup_{k \le s_{i+1}} \Dec(V_{i+1,k},\pi_i),$$ 
which we rewrite as $$\Dec(V,i,<)=\{V_{i,1},\dots,V_{i,s_{i}}\}.$$
An easy
decreasing induction proves that for $i=1,\dots,n$ and $k \le s_i$,
every $V_{i,k}$ is $\pi_j$-equiprojectable for $j=i,\dots,n$:
\begin{itemize}
\item For $i=n$, $\Dec(V,n,<)$ is simply $\{V\}$, which is
  $\pi_n$-equiprojectable (since $\pi_n$ is the identity).
\item For $i<n$, assuming that the claim holds for $\Dec(V,i+1,<)$, we
  prove it for $\Dec(V,i,<)$. To do so, it is enough to take
  $V_{i+1,k}$ in $\Dec(V,i+1,<)$ and prove that every $V'$ in
  $\Dec(V_{i+1,k},\pi_i)$ is $\pi_j$-equiprojectable, for
  $j=i,\dots,n$.

  Obviously, $V'$ is $\pi_i$-equiprojectable. Besides, since by the
  induction assumption $V_{i+1,k}$ is $\pi_j$-equiprojectable for
  $j=i+1,\dots,n$, Lemma~\ref{lemma:comp} implies that $V'$ is also
  $\pi_j$-equiprojectable for $j=i+1,\dots,n$.
\end{itemize}
Taking $i=1$, $\Dec(V,1,<)$ is the {\em equiprojectable decomposition}
of $V$; we will actually denote it by $\Dec(V,<)$. Dropping the
subscript ${}_1$, we will write
$$\Dec(V,<)=\{V_{1},\dots,V_{s}\}.$$ This is thus a decomposition of $V$
into pairwise disjoint equiprojectable sets $V_j$.

Aubry and Valibouze proved in~\cite{AuVa00} that an algebraic set is
equiprojectable if and only if its defining ideal is generated by a
triangular set. Besides, by Lemma~\ref{lemma:def}, each $V_j$ is
defined over $\K$; thus, its defining ideal is generated by a
triangular set $\Tt^{(j)}$ in $\K[\X]$. As said in the introduction,
we will write $\Dr(V,<)$ to denote the collection of the triangular
sets $\{\Tt^{(1)},\dots,\Tt^{(s)}\}$. In ideal-theoretic terms, the
ideals $\langle \Tt^{(j)}\rangle$ are thus pairwise coprime, and their
intersection is the defining ideal $I$ of $V$, so that $\K[\X]/I\simeq 
R_{\Tt^{(1)}} \times \cdots \times R_{\Tt^{(s)}}$.

%% When there can be no ambiguity on the variable order, we will omit
%% it and write simply $\Dec(V)$ and $\Dr(V)$.

\medskip

The following proposition gives a cost estimate on the computation of
the equiprojectable decomposition, using a univariate representation
as input.
\begin{Prop}\label{lemma:equi}
  Let $V\subset \Kbar^n$ be a zero-dimensional algebraic set defined
  over $\K$, of degree $\delta$. If the characteristic of $\K$ is
  equal to $0$ or greater than $\delta^2$, given a univariate
  representation $\Ur$ of $V$, we can compute
  $\Dr(V,<)=\{\Tt^{(1)},\dots,\Tt^{(s)}\}$ in expected time
  $O(n\CC(\delta) (n+\log(\delta))).$ Besides, the following change of
  bases can be done in time $O(n\CC(\delta))$:
  \begin{itemize}
  \item given $A$ in $\K[X]/\langle P \rangle$, compute its images
    $(A_1,\dots,A_s)$ in $R_{\Tt^{(1)}} \times \cdots \times
    R_{\Tt^{(s)}}$;
  \item given $(A_1,\dots,A_s)$ in $R_{\Tt^{(1)}}\times\cdots\times
    R_{\Tt^{(s)}}$, compute their preimage $A$ in $\K[X]/\langle P
    \rangle$.
  \end{itemize}
\end{Prop}
\begin{proof}
  Let us write as before $\Dec(V,<)=\{V_{1},\dots,V_{s}\}$.  The
  algorithm to compute $\Dr(V,<)$ proceeds in two steps: first, we
  compute univariate representations of all $V_j$; secondly, we
  convert them into triangular sets. As we go, we also explain how to
  perform the change of basis from $A$ to $(A_1,\dots,A_s)$, and back.

  \paragraph{Step 1.} Recall the definition of the sequence
  $\Dec(V,i,<)$: we have $\Dec(V,n,<) = \{V\}$ and starting from
    $$\Dec(V,i+1,<)=\{V_{i+1,1},\dots,V_{i+1,s_{i+1}}\},$$
    we set
    $$\Dec(V,i,<) =\cup_{k \le s_{i+1}} \Dec(V_{i+1,k},\pi_i).$$ The
    first step of the algorithm follows the same loop, and computes
    univariate representations of all $V_{i,k}$. We set
    $\Ur_{n,1}=\Ur$, and for $i=n-1,\dots,1$, we let
    $\Ur_{i,1},\dots,\Ur_{i,s_i}$ be the univariate representations
    obtained by applying the algorithm of Proposition~\ref{prop:dec}
    to $\Ur_{i+1,1},\dots,\Ur_{i,s_{i+1}}$ and $\pi_i$. If
    $\delta_{i+1,k}$ denotes the degree of $V_{i+1,k}$, applying the
    algorithm of Proposition~\ref{prop:dec} to $\Ur_{i+1,k}$ and
    $\pi_i$ takes an expected time
    \[
    O(\CC(\delta_{i+1,k})(n+\log(\delta_{i+1,k})) ).
    \]
    Using the super-linearity of $\CC$, and the fact that
    $\delta_{i+1,1}+\cdots+\delta_{i+1,s_{i+1}} = \delta,$ the time
    spent at index $i$ is seen to be an expected
    $O(\CC(\delta)(n+\log(\delta))).$ Summing over all $i$, the total
    time is an expected
    \[
    O(n\CC(\delta)(n+\log(\delta))).
    \]

    Let $P$ be the characteristic polynomial of $\Ur$, and let
    $P_1,\dots,P_{s_1}$ be those of
    $\Ur_{1,1},\dots,\Ur_{1,s_{1}}$. Since the separating elements of
    $\Ur$ and $\Ur_{1,1},\dots,\Ur_{1,s_{1}}$ are the same, we have
    $P=P_1\dots P_{s_1}$. The change of basis $\K[X]/\langle P\rangle
    \to \K[X]/\langle P_1\rangle\times \cdots\times \K[X]/\langle
    P_{s_1}\rangle$ is done by multiple reduction, and the inverse
    conversion is done using the Chinese Remainder Theorem. Using the
    results of~\cite[Chapter~10]{GaGe03}, both conversions take time
    $O(\M(\delta)\log(\delta))$, which is $O(\CC(\delta))$.

    \paragraph{Step 2.} Starting from $\Ur_{1,1},\dots,\Ur_{1,s_{1}}$,
    we now compute the corresponding triangular sets
    $\Tt^{(1)},\dots,\Tt^{(s)}$. This is done by applying
    Lemma~\ref{lemma:conv}, which shows that we can compute each
    triangular set $\Tt^{(j)}$ in expected time $O(n^2\CC(\delta_j))$,
    where $\delta_j$ is the degree of $V_j$. Summing over all $j$ and
    using the super-linearity of the function $\CC$ gives a total
    expected time of $O(n^2\CC(\delta))$.

    Using the notation of Subsection~\ref{ssec:uni}, the conversion
    $$\K[X]/\langle P_1\rangle\times \cdots\times \K[X]/\langle
    P_{s_1}\rangle \to R_{\Tt^{(1)}} \times \cdots \times
    R_{\Tt^{(s)}}$$ and its inverse are done by applying
    $$(\Phi_{\Tt^{(1)},\Ur_{1,1}},\dots,\Phi_{\Tt^{(s_1)},\Ur_{1,{s_1}}})
    \quad\text{and}\quad
    (\Psi_{\Tt^{(1)},\Ur_{1,1}},\dots,\Psi_{\Tt^{(s_1)},\Ur_{1,{s_1}}}).$$
    By Lemma~\ref{lemma:conv}, and using the super-linearity of
    $\CC$, each conversion takes time $O(n\CC(\delta))$.
  \end{proof}

%%%%%%%%%%%%%%%%%%%%%%%%%%%%%%%%%%%%%%%%%%%%%%%%%%%%%%%%%%%%

  \subsection{Solving question ${\bf P}_1$}\label{ssec:P1}

  We can now show how to solve question ${\bf P}_1$ stated in the
  introduction.  Given triangular sets $\Tt^{(1)},\dots,\Tt^{(\ell)}$
  and $\Ss^{(1)},\dots,\Ss^{(r)}$ for an order $<$, and a target order
  $<'$, we want to compute $\Dr (V, <'),$ with
$$V=
V(\Tt^{(1)}) \cup \cdots \cup V(\Tt^{(\ell)}) - V(\Ss^{(1)}) - \cdots
- V(\Ss^{(r)}).$$ We let $\delta$ be the sum of the degrees of
$\Tt^{(1)},\dots,\Tt^{(\ell)}$ and $\Ss^{(1)},\dots,\Ss^{(r)}$ and we
make the assumption that the characteristic of $\K$ is equal to $0$ or
greater than $\delta^2$.

Our strategy is to reduce to univariate representations, perform the
set theoretic operations on univariate polynomials, and finally
compute the equiprojectable decomposition for the new order.

\paragraph{Step 1.} We compute univariate representations
$\Ur_1,\dots,\Ur_\ell$ and $\Vr_1,\dots,\Vr_r$ of respectively
$V(\Tt^{(1)}),\dots,V(\Tt^{(\ell)})$ and
$V(\Ss^{(1)}),\dots,V(\Ss^{(r)})$. By Lemma~\ref{lemma:conv}, this can
be done in expected time
$$O\big (n^2(\CC(\delta_1) +\cdots + \CC(\delta_\ell)+ \CC(\delta'_1)
+\cdots + \CC(\delta'_r))\big ),$$ where $\delta_i$ is the degree of
$\Tt^{(i)}$ and $\delta'_i$ is the degree of $\Ss^{(i)}$.  Using the
super-linearity of $\CC$, this is seen to be an expected $O(n^2
\CC(\delta))$.

\paragraph{Step 2.} We compute univariate representations $\Ur$ of
$V(\Tt^{(1)}) \cup \dots \cup V(\Tt^{(\ell)})$ and $\Vr$ of
$V(\Ss^{(1)}) \cup \dots \cup V(\Ss^{(r)})$. The following
divide-and-conquer process takes an expected time
$O(n\CC(\delta)\log(\delta))$ to achieve this task.

We apply repeatedly the union algorithm of Lemma~\ref{lemma:merge} to
$\Ur_1,\dots,\Ur_\ell$, respectively $\Vr_1,\dots,\Vr_r$. To compute
say $\Ur$, we let $\ell'=\lceil \ell/2\rceil$, and we compute
recursively univariate representations of
$$V(\Tt^{(1)}) \cup \dots \cup V(\Tt^{(\ell')}) \quad\text{and}\quad
V(\Tt^{(\ell'+1)}) \cup \dots \cup V(\Tt^{(\ell)});$$ then, these two
univariate representations are merged by means of
Lemma~\ref{lemma:merge}. The running time analysis is the same as in
the proof of Proposition~\ref{prop:dec}: the divide-and-conquer
structure of the algorithm induces the loss of a logarithmic factor,
as is the case for other algorithms with the same
structure~\cite[Chapter~10]{GaGe03}.

\paragraph{Step 3.} By another application of Lemma~\ref{lemma:merge}
to $\Ur$ and $\Vr$, this time for computing a set-theoretic
difference, we finally obtain a univariate representation $\Wr$ of
$V$. This takes an expected time $O(n\CC(\delta))$.

\paragraph{Step 4.} Starting from $\Wr$, we compute $\Dr(V,<')$ using
the algorithm of Proposition~\ref{lemma:equi}.  This takes an expected
time $O(n \CC(\delta)(n+\log(\delta))).$

\medskip

The total cost of this algorithm is an expected $O(n
\CC(\delta)(n+\log(\delta)))$, as claimed in Theorem~\ref{theo:CtoE}.

%%%%%%%%%%%%%%%%%%%%%%%%%%%%%%%%%%%%%%%%%%%%%%%%%%%%%%%%%%%%

\subsection{Solving question ${\bf P}_2$}\label{ssec:P2}

Next, we show how to solve question ${\bf P}_2$ stated in the
introduction. Given a triangular set $\Tt$ in $\K[X_1,\dots,X_n]$, for
a variable order $<$, as well as $F$ in $R_\Tt$ and a target variable
order $<'$, we are to compute the equiprojectable decompositions
$$\Dr(V(\Tt)\cap V(F),<') \quad\text{and}\quad \Dr(V(\Tt)- V(F),<'),$$
as well as the inverse of $F$ modulo each $\Tt'$ in $\Dr(V(\Tt)-
V(F),<')$. We let $\delta$ be the degrees of $\Tt$ and we make the
assumption that the characteristic of $\K$ is equal to $0$ or greater
than~$\delta^2$.

Our strategy is similar to the one of the previous subsection: we
convert to a univariate representation, operate with univariate
polynomials, and convert back to triangular representations.

\paragraph{Step 1.} We compute a univariate representation
$\Ur=(P,\Uu,\mu)$ of $V(\Tt)$ and $F^\star=\Psi_{\Tt,\Ur}(F)$. By
Lemma~\ref{lemma:conv}, this can be done in expected time $O(n^2
\CC(\delta))$.

\paragraph{Step 2.} We compute $P'=\gcd(P,F^\star)$ and $P''=P/P'$,
as well as the inverse $G^\star$ of $F^\star$ modulo $P''$ (this
inverse exists, since $P$ is squarefree). This takes time
$O(\M(\delta)\log(\delta))$, which is $O(\CC(\delta))$.

The roots of $P'$ describe the points of $V(\Tt)$ where $F$ vanishes;
the roots of $P''$ describe those where $F$ is nonzero.

\paragraph{Step 3.} Writing $\Uu=(U_1,\dots,U_n)$, we compute $U'_i =
U_i \bmod P'$ and $U''_i=U_i \bmod P''$ for all $i$, and we define
$\Ur'=(P',(U'_1,\dots,U'_n),\mu)$ and
$\Ur''=(P'',(U''_1,\dots,U''_n),\mu)$. This takes time
$O(n\M(\delta))$, which is negligible compared to the cost of Step~1.

Note that $\Ur'$ is a univariate representation of $V(\Tt)\cap V(F)$,
and that $\Ur''$ is a univariate representation of $V(\Tt)- V(F)$.

\paragraph{Step 4.} Starting from $\Ur'$ and $\Ur''$ we compute
the equiprojectable decompositions $\Dr(V(\Tt)\cap V(F),<')$ and
$\Dr(V(\Tt)- V(F),<')$ using the algorithm of
Proposition~\ref{lemma:equi}.  This takes an expected time
$O(n\CC(\delta)(n+\log(\delta))).$ Besides, using the second part of
Proposition~\ref{lemma:equi}, we can compute the image of $G^\star$ in
each $R_{\Tt'}$, for $\Tt'$ in $\Dr(V(\Tt)- V(F),<')$. This image is
the inverse of $F$ in $R_{\Tt'}$.

\medskip

As for question ${\bf P}_2$, the total cost of this algorithm is an
expected $O(n\CC(\delta)(n+\log(\delta)))$, as claimed in
Theorem~\ref{theo:CtoE}.

%%%%%%%%%%%%%%%%%%%%%%%%%%%%%%%%%%%%%%%%%%%%%%%%%%%%%%%%%%%%

\subsection{Experimental results}

This section reports on experimental results obtained with a Maple
implementation of the algorithms of Subsection~\ref{ssec:P1}
and~\ref{ssec:P2}.

Our implementation supports inputs with coefficients in finite fields
of the form $\F_p$, $p$ prime. This is the most natural choice, since
over base fields such as $\Q$ or rational function fields, the cost of
arithmetic operations in the base field cannot be assumed to be
constant. For inputs defined over e.g. $\Q$, the natural approach
would be to use modular methods, using for instance lifting techniques
(for which the equiprojectable decomposition is particularly well
suited, as we pointed out in the introduction).

Over base fields such as $\F_p$, we have two choices for modular
composition and power projection: algorithms following Brent and
Kung's idea, as described in Section~\ref{ssec:basics}, or the
extension of the Kedlaya-Umans algorithm given in~\cite{PoSc10}.
Unfortunately, even though the latter is asymptotically better, the
large constants hidden in the $\Ot$ notation make it inferior for the
range of degrees we consider. Thus, our implementation relies on the
Brent-Kung approach.

Other than modular composition and power projection, our algorithms
use only univariate and bivariate polynomial arithmetic. As a result,
they were implemented using the {\tt modp1} functions, which provide
fast implementations of arithmetic operations in $\F_p[X]$, for $p$ a
word-size prime.

The following timings are obtained using Maple 15 on an 2.8 GHz AMD
Athlon II X2 240e processor. The base field is $\F_p$, with
$p=962592769$. All timings are in seconds, and all computations were
interrupted whenever they used 2Gb of RAM or more.

Our first experiments concern the particular case of question ${\bf P}_1$,
where the input and the target order are the same, and $r=0$. In other
words, we take as input some triangular sets
$\Tt^{(1)},\dots,\Tt^{(\ell)}$ for an order $<$, and we compute the
equiprojectable decomposition of
$$V(\Tt^{(1)}) \cup \cdots \cup V(\Tt^{(\ell)}),$$ for the same order.
In Table~\ref{table1}, we shows comparisons with the function {\tt
  Equi\-projectableDe\-com\-position} of the {\tt RegularChains}
library~\cite{LeMoXi05}, which has similar specifications
(we are not aware of other implementations of such an algorithm).

In each sub-table, the number $n$ of variables is fixed; we show
timings for the equiprojectable decompositions of sets of points of
cardinality $\delta$; the column $d$ gives an upper bound on all $d_i$
that appear as main degrees in the triangular sets in the output.  In
almost all cases, our implementation does better than the built-in
function; the fact that we are relying on the {\tt modp1} functions is
certainly a key factor for this.

\begin{table}
  \caption{Timings for equiprojectable decomposition}\label{table1}
  \begin{center}
  \begin{tabular}{cc}
  \begin{tabular}{c|c|c||c|c}
    $n$ & $d$ & $\delta$ & us & Maple \\ \hline
    3& 2& 4&  0.03 & 0.03 \\
    3& 3& 10& {0.07}& 0.12\\
    3& 4& 20& {0.12} & 0.52 \\
    3& 5& 35& {0.22} &1.6 \\
    3& 6& 56& {0.44} & 4.2
  \end{tabular}
&
  \begin{tabular}{c|c|c||c|c}
    $n$ & $d$ & $\delta$ & us & Maple \\ \hline
4& 2& 5 & 0.06 & {0.05} \\
4& 3& 15 & {0.2} & 0.4\\
4& 4& 35 & {0.3} & 2.1\\
4& 5& 70 & {0.8} & 8.4\\
4& 6& 126 & {1.9} & 40
  \end{tabular} 
\\
\\
  \begin{tabular}{c|c|c||c|c}
    $n$ & $d$ & $\delta$ & us & Maple \\ \hline
5& 2& 6 & 0.09 & {0.08}\\
5& 3& 21 & {0.37} & 0.96\\
5& 4& 56 & {0.81} & 6.5\\
5& 5& 126& {2.4} & 45 \\
5& 6& 252& {9.5} & 512
  \end{tabular}
&
  \begin{tabular}{c|c|c||c|c}
    $n$ & $d$ & $\delta$ & us & Maple \\ \hline
6& 2& 7 & 0.15 & {0.13} \\
6& 3& 28& {0.5} & 2.1\\
6& 4& 84& {1.8} &19\\
6& 5& 210& {8.2} & 300\\
6& 6& 462& {49} & 5885
  \end{tabular} 
  \end{tabular}
  \end{center}
\end{table}

Our second experiments address inverse computation modulo a triangular
set, which is a particular case of question ${\bf P}_2$: the input and
the target order are the same, and (by construction of our examples),
no splitting occurred. In other words, we take as input a triangular
set $\Tt$ and $F\in R_\Tt$, invertible in $R_\Tt$; we output the
inverse of $F$ in $R_\Tt$.

In Table~\ref{table2}, we give examples for various situations: 
$n$ denotes the number of variables and $d$ is such that 
the input triangular set has multidegree $(d,\dots,d)$, of length $n$;
thus, its degree $\delta$ is $d^n$. 

We show comparisons with the function {\tt Inverse} of the {\tt
  RegularChains} library. This function may induce splittings; if we
wanted the same output as in our implementation, we would also have to
perform a recombination after the call to {\tt Inverse} (we did not
include this step in the timings). As in the previous example, our
code usually does better.

We also include timings obtained by using the C {\tt modpn}
library~\cite{LiMoRaSc09}, which can be called from a Maple
session. Obviously, we expect this compiled library to be much faster
than our interpreted code; however, timings are sometimes within a
factor of 10 or less, which we see as a sign that our implementation
performs well. Note that {\tt modpn} relies on FFT techniques, as a
result, only those finite fields $\F_p$ with suitable roots of unity
are supported (the field $\F_p$ in our examples is one of them).

\begin{table}
  \caption{Timings for inversion in $R_\Tt$}\label{table2}
  \begin{center}
  \begin{tabular}{cc}
  \begin{tabular}{c|c|c||c|c|c}
    $n$ & $d$ & $\delta$ & us & {\tt Inverse} & {\tt modpn} \\ \hline
3& 2& 8 & 0.04 & 0.3 & 0.01\\
3& 3& 27 & 0.06 & 1.4 & 0.01\\
3& 4& 64 & 0.14 & 5.2 & 0.02\\
3& 5& 125& 0.24& 6.1 & 0.05\\
3& 6& 216& 0.75& 21 & 0.06
  \end{tabular}
&
  \begin{tabular}{c|c|c||c|c|c}
    $n$ & $d$ & $\delta$ & us & {\tt Inverse} & {\tt modpn} \\ \hline
4& 2& 16& 0.07& 1.1& 0.01\\
4& 3& 81& 0.2 & 4.8 &0.06\\
4& 4& 256& 1& 600 & 0.1\\
4& 5& 625& 5.3& 10536& 0.8\\
4& 6& 1296& 23& $> 2$ Gb& 1.2
  \end{tabular} 
\\
\\
  \begin{tabular}{c|c|c||c|c|c}
    $n$ & $d$ & $\delta$ & us & {\tt Inverse} & {\tt modpn} \\ \hline
5& 2& 32  & 0.14 & 210 & 0.03\\
5& 3& 243  &1 & 1576 & 0.42\\
5& 4& 1024 &1.5& $> 2$ Gb &1.2\\
5& 5& 3125 &151 & $> 2$  Gb & 24\\
5& 6& 7776 &1007 & $> 2$  Gb & 37
  \end{tabular}
&
  \begin{tabular}{c|c|c||c|c|c}
    $n$ & $d$ & $\delta$ & us & {\tt Inverse} & {\tt modpn} \\ \hline
6& 2& 64 & 0.3 & $> 2$ Gb & 0.1\\
6& 3& 729 & 8.8 & $> 2$ Gb & 4.6\\
6& 4& 4096 & 273 & $> 2$ Gb & 18\\
6& 5& 15625 & 5099 & $> 2$ Gb &661\\
6& 6& 46656 & 67339 & $> 2$ Gb& 1135
  \end{tabular} 
  \end{tabular}
  \end{center}
\end{table}

%%%%%%%%%%%%%%%%%%%%%%%%%%%%%%%%%%%%%%%%%%%%%%%%%%%%%%%%%%%%
%%%%%%%%%%%%%%%%%%%%%%%%%%%%%%%%%%%%%%%%%%%%%%%%%%%%%%%%%%%%
%%%%%%%%%%%%%%%%%%%%%%%%%%%%%%%%%%%%%%%%%%%%%%%%%%%%%%%%%%%%

\section{The converse reduction}\label{sec:EtoC}

This section is mostly independent from the other ones. In the
previous sections, we used modular composition and power projection as
our basic subroutines, and reduced other questions to these two
operations. In this section, we will do the opposite, by reducing
modular composition and power projection to equiprojectable
decomposition.

%% %% \footnote{Par hasard : c'est une raison suppl\'ementaire
%% %%   pour laquelle la projection des puissances bivari\'ee est
%% %%   d\'etaill\'ee dans la section \ref{ssec:basics} ?}
%% Euh non, pas vraiment..

As mentioned in the introduction, modular composition and power
projection are dual problems. An algorithmic theorem called the {\em
  transposition principle} shows that an algorithm for the former can
be transformed into an algorithm for the latter, and
conversely~\cite{BuClSh97,BoLeSc03}: this result could in principle
allow us to deal only with e.g. modular composition. However, it
applies only in a restricted computational model (using {\em linear
  programs}), which is not suited to questions such as decompositions
of triangular sets (which are inherently non-linear). As a result, we
give explicit reductions for both modular composition and power
projection.

In the introduction, we defined $\EE: \N^2 \to \N$ as a function such
that one can solve problem ${\bf P}_1$ (computing the equiprojectable
decomposition of a family of triangular sets in $n$ variables, with
sum of degrees $\delta$) using $\EE(n,\delta)$ base field operations.

Recall then the statement of Theorem~\ref{theo:EtoC}: we take
$(m,n)=(1,1)$ or $(m,n)=(1,2)$, and we let $\Tt$ be a triangular set
in $n$ variables that generates a radical ideal. Then, we can compute
modular compositions and power projections modulo $\langle \Tt
\rangle$ with parameters $(m,n)$ and size $\delta_{\bf f} \le
\delta_\Tt$ in time $2\EE(4,\delta_\Tt)+\Ot(\delta_\Tt)$.

The two subsections address respectively modular composition and power
projection. In both cases, we can assume that $n=2$, since any
triangular set in one variable (that is, any polynomial $T_1(X_1)$)
can be seen as a triangular set in two variables, by adding a dummy
polynomial $T_2(X_1,X_2)=X_2$. Note that the proofs would generalize
to computations in more than two variables, and would involve terms of
the form $\EE(n+2,\delta_\Tt)$.

%%%%%%%%%%%%%%%%%%%%%%%%%%%%%%%%%%%%%%%%%%%%%%%%%%%%%%%%%%%%

\subsection{Modular composition}

%% \footnote{Question qui m'est venue \`a
%%   l'esprit \`a la lecture de ce paragraphe : si on utilise $n>2$
%%   variables, la strat\'egie que tu d\'ecris l\`a ne nous donnerait pas
%%   du $2\EE(n+2,\delta_\Tt)+\Ot(\delta_\Tt)$ ? Si oui, ca vaudrait pas
%%   le coup de mettre une remarque \`a ce sujet (un truc du genre ``on
%%   donne pas la preuve, mais ca se g\'en\'eralise en\dots'')}

Following the previous discussion let thus $\Tt=(T_1,T_2)$ be a
triangular set in $\K[X_1,X_2]$, $G$ in $R_\Tt$, and $F$ in $\K[Y]$,
of degree $\deg(F) \le \delta_\Tt$. We show here how to compute
$K=F(G) \in R_\Tt$, using change of order as our main subroutine.

Consider the triangular set (for the order $X_1 < X_2< Y$)
\[
\Tt' \left |
  \begin{array}{l}
    Y-G(X_1,X_2)\\
    T_2(X_1,X_2)\\
    T_1(X_1);
  \end{array}\right .
\]
let $V\subset \Kbar^3$ be its zero-set, and let us compute $\Dr(V,
<')$, where $<'$ is the order $Y<'X_1<'X_2$. We obtain a family of
triangular sets $\Uu^{(1)},\dots,\Uu^{(N)}$ of the form
\[
\Uu^{(i)} \left |
  \begin{array}{l}
    U_{i,2}(Y,X_1,X_2)\\
    U_{i,1}(Y,X_1)\\
    R_i(Y).
  \end{array}\right .
\]
Let now $I$ be the ideal generated by the polynomials (which do not
form a triangular set, since the first polynomial is not reduced)
\[
\left |
  \begin{array}{l}
    Z-F(Y)\\
    Y-G(X_1,X_2)\\
    T_2(X_1,X_2)\\
    T_1(X_1).
  \end{array}\right .
\]
After reduction, we see that $I$ is generated by the triangular set
(for the order $X_1 <X_2 < Y <Z$)
\[
\Tt'' \left |
  \begin{array}{l}
    Z-K(X_1,X_2)\\
    Y-G(X_1,X_2)\\
    T_2(X_1,X_2)\\
    T_1(X_1),
  \end{array}\right .
\]
where $K$ is the polynomial we want to compute. On the other hand, the
construction of the triangular sets $\Uu^{(i)}$ shows that $I$ is the
intersection of the ideals
generated by the triangular sets $\Vv^{(i)}$ (for the order $Y <' X_1
<' X_2 <' Z$) given by
\[
\Vv^{(i)} \left |
  \begin{array}{l}
    Z-F_i(Y)\\
    U_{i,2}(Y,X_1,X_2)\\
    U_{i,1}(Y,X_1)\\
    R_i(Y),
  \end{array}\right .
\]
with $F_i = F \bmod R_i$. The algorithm is then the following:
\begin{itemize}
\item First, we compute all triangular sets $\Uu^{(i)}$. Since $\Tt'$
  generates a radical ideal, this can be done in $\EE(3,\delta_\Tt)\le
  \EE(4,\delta_\Tt)$ base field operations (obviously, $\EE(n,\delta)
  \le \EE(n',\delta)$ holds for all $n\le n'$, as can be seen by using
  $n'-n$ dummy polynomials to obtain a triangular set in $n'$
  variables).
\item Next, we compute all triangular sets $\Vv^{(i)}$. This requires
  us to compute all $F_i$. Since $\deg(F) \le \delta_\Tt$, and since
  the sum of the degrees of the $R_i$ is at most $\delta_\Tt$ as well,
  all $F_i$ can be computed in time
  $O(\M(\delta_\Tt)\log(\delta_\Tt))$ using fast multiple
  reduction~\cite[Chapter~10]{GaGe03}.
\item Finally, we compute $\Tt''$, and thus $K$, by computing the
  equiprojectable decomposition of $V(\Vv^{(1)}) \cup \cdots \cup
  V(\Vv^{(N)})$, for the order $X_1 <X_2 < Y <Z$. Again, this takes
  time $\EE(4,\delta_\Tt)$.
\end{itemize}
The total time is at most
$2\EE(4,\delta_\Tt)+O(\M(\delta_\Tt)\log(\delta_\Tt))$, which fits
into the claimed bound.

%%%%%%%%%%%%%%%%%%%%%%%%%%%%%%%%%%%%%%%%%%%%%%%%%%%%%%%%%%%%

\subsection{Power projection}

We will now prove the second part of Theorem~\ref{theo:EtoC}, dealing
with power projection. Let thus $\Tt=(T_1,T_2)$ be a triangular set in
$\K[X_1,X_2]$ that generates a radical ideal, let $G$ be in $R_\Tt$,
and let $\ell:R_\Tt \to \K$ be a $\K$-linear form. Given an integer $f
\le \delta_\Tt$, we show here how to compute the values $\ell(G^c)$,
for $0 \le c< f$. We start with a folklore lemma involving univariate
computations only.

\paragraph{Univariate computations.}
Let $\A$ be a ring, $F$ a monic polynomial of degree $d$ in $\A[X]$,
and $R$ the free $\A$-module $\A[X]/\langle F \rangle$, with the
(classes of) $1,X,\dots,X^{d-1}$ as a basis. In this context, the {\em
  trace} $\tau: R \to \A$ is still well-defined, with $\tau(A)$ being
the trace of the multiplication map by $A$ in~$R$.  For $A\in R$ and
$\ell$ an $\A$-linear form $R \to \A$, the $\A$-linear form $A\cdot
\ell$ is defined as before, by $(A \cdot \ell)(B)=\ell(AB)$.

\begin{Lemma}\label{lemma:trdeg}
  Suppose that the derivative $\partial F/\partial X$ of $F$ is
  invertible in $R$, with inverse $G$. Given $G$, and given an
  $\A$-linear form $\ell: R \to \A$, we can compute $A$ in $R$ such
  that $\ell=A \cdot \tau$, using $O(\M(d))$ operations in $\A$.
\end{Lemma}
\begin{proof}
  Let us define another useful $\A$-linear form, the {\em residue}
  $\rho: R \to \A$, by $\rho(X^i)=0$ for $i < d-1$ and
  $\rho(X^{d-1})=1$. Given $\ell$ as above, it is known that there
  exists $B$ such that $\ell=B \cdot \rho$. Indeed, a straightforward
  computation shows that the values $(B \cdot \rho)(X^i)$, for
  $i=0,\dots,d-1$, are the coefficients of ${\rm rev}(B,d-1)/{\rm
    rev}(F,d) \bmod X^{d}$, where for any polynomial $P\in \A[X]$ and
  any $d \ge \deg(P)$, we write ${\rm rev}(P,d)=X^d P(1/X)$. This
  implies that given $\ell$, we can find the requested $B$ by means of
  a power series multiplication modulo $X^d$, which can be done in
  $\M(d)$ operations in $\A$.
  
  Furthermore, the {\em Euler formula}~\cite[Proposition
    2.4]{Demazure87} shows that $\tau = \partial F/\partial X \cdot
  \rho$, so that $\rho = G \cdot \tau$. With $\ell$ and $B$ as above,
  this implies that we have $\ell = A \cdot \tau$, with $A=B G \bmod
  F$. Computing $A$ thus takes another $O(\M(d))$ operations in $\A$,
  proving the lemma.
\end{proof}

\paragraph{Bivariate computations.} We will now apply the results 
of the former paragraph in a bivariate context. The notation is the
one introduced at the beginning of this subsection; furthermore, we
let ${\rm tr}:R_\Tt \to \K$ be the trace linear form. We also write
$d_1=\deg(T_1,X_1)$ and $d_2=\deg(T_2,X_2)$, so that $\delta_\Tt=d_1
d_2$.

\begin{Lemma}\label{lemma:trdeg2}
  Given a $\K$-linear form $\ell:R_\Tt \to \K$, one can compute an
  element $A \in R_\Tt$ such that $\ell = A\cdot {\rm tr}$ in time
  $O(\M(d_1)\M(d_2)\log(d_1)\log(d_2))$.
\end{Lemma}
\begin{proof}
  Let us define $S_\Tt = \K[X_1]/\langle T_1\rangle$, so that we have
  $R_\Tt =S_\Tt[X_2]/\langle T_2\rangle$. Let further $\tau_1:S_\Tt\to
  \K$ and $\tau_2:R_\Tt \to S_\Tt$ be the trace forms; thus, $\tau_1$
  is $\K$-linear, $\tau_2$ is $S_\Tt$-linear, and we have ${\rm
    tr}=\tau_1 \circ \tau_2$.

  First, we are going to factor $\ell: R_\Tt \to \K$ as $\ell =
  \tau_1 \circ L$, where $L:R_\Tt \to S_\Tt$ is a suitable
  $S_\Tt$-linear form. Computing $L$ amounts to compute
  $\lambda_{i_2}=L(X_2^{i_2})$, for $i_2=0,\dots,d_2-1$; the condition
  defining $L$ is equivalent to $\ell(X_1^{i_1}
  X_2^{i_2})=\tau_1(L(X_1^{i_1} X_2^{i_2}))$, for $i_1=0,\dots,d_1-1$
  and $i_2=0,\dots,d_2-1$. This can be rewritten as $\ell(X_1^{i_1}
  X_2^{i_2})=\tau_1(X_1^{i_1}\lambda_{i_2})$, by $S_\Tt$-linearity
  of~$L$. For a fixed $i_2 < d_2$, let $\ell_{i_2}$ be the $\K$-linear
  form $S_\Tt \to \K$ defined by $\ell_{i_2}(A)=\ell(A X_2^{i_2})$.
  Then, the previous condition says that $\ell_{i_2} = \lambda_{i_2}
  \cdot \tau_1$.

  Computing the linear forms $\ell_{i_2}$ is free (since their values
  on the canonical basis of $S_\Tt$ are simply values of $\ell$);
  then, finding $\lambda_{i_2}$ is done by first inverting $T_1'$
  modulo $T_1$, and applying Lemma~\ref{lemma:trdeg} for the extension
  $S_\Tt \to \K$. The total time to computing all $\lambda_{i_2}$ is
  thus $O((\log(d_1)+d_2) \M(d_1))$.

  Now that we have written $\ell = \tau_1 \circ L$, we will apply
  Lemma~\ref{lemma:trdeg} to $L$, for the extension $R_\Tt \to S_\Tt$.
  This requires us to invert $\partial T_2/\partial X_2$ in $R_\Tt$; a
  quasi-linear time algorithm is given in~\cite{AcCoMa03}, with a cost
  $O(\M(d_1)\M(d_2) \log(d_1)\log(d_2))$. Once this is done,
  Lemma~\ref{lemma:trdeg} gives us an element $A \in R_\Tt$ such that
  $L=A\cdot \tau_2$ in time $O(\M(d_1)\M(d_2))$.

  To summarize, we have written $\ell = \tau_1 \circ L$ and $L=A\cdot
  \tau_2$, so that $\ell(B) = \tau_1 (\tau_2(AB))$ holds for all $B
  \in R_\Tt$. Since $\tau_1 \circ \tau_2 = {\rm tr}$, this implies that 
  $\ell = A \cdot {\rm tr}$.
\end{proof}

\paragraph{Transposed multiple reduction.} Our next ingredient 
is an algorithm for the following operation. Consider some pairwise
coprime monic polynomials $R_1,\dots,R_N$ in $\K[X]$, 
and let $R=R_1 \cdots R_N$.

We have already mentioned the multiple reduction map $\K[X]/\langle R
\rangle \to \K[X]/\langle R_1 \rangle \times \cdots \times
\K[X]/\langle R_N \rangle$; writing $d=\deg(R)$, this operation can be
done in time $O(\M(d)\log(d))$. In this paragraph, we will discuss the
dual map. On input linear forms $\ell_i:\K[X]/\langle R_i \rangle\to
\K$, this dual map computes the linear form $\ell:\K[X]/\langle R
\rangle$ defined by
$$A\mapsto \sum_{i\le N} \ell_i(A \bmod R_i),$$ where all $\ell_i$ and
$\ell$ are given by means of their values on the monomials bases of
the respective $\K[X]/\langle R_i \rangle$ and $\K[X]/\langle R
\rangle$. In other words, it computes the values
$$ \sum_{i\le N} \ell_i(X^j \bmod R_i),$$ for $j=0,\dots,d-1$.
In~\cite{BoLeSaScWi04}, an algorithm called ${\sf TSimulMod}$ is given
that solves this problem in time $O(\M(d)\log(d))$. Computing the
above values up to index $e$, for some $e>d$, can then be done in time
$O(\M(e))$, see for instance~\cite{BoLeSc03}.

\paragraph{Conclusion.} Let us return to the proof of Theorem~\ref{theo:EtoC}.
On input $\Tt=(T_1,T_2)$, $G\in R_\Tt$ and $\ell:R_\Tt\to \K$, we will
show how to compute the values $\ell(G^c)$, for $0 \le c<\delta_\Tt$.
Using the algorithm of Lemma~\ref{lemma:trdeg2}, we can compute $A \in
R_\Tt$ such the values we want are of the form ${\rm tr}(A G^c)$, for
$0 \le c<\delta_\Tt$.

Let us introduce the triangular set (for the order $X_1<X_2<Y<Z$)
\[
\Tt' \left | 
  \begin{array}{l}
    Z-G(X_1,X_2)\\
    Y-A(X_1,X_2)\\
    T_2(X_1,X_2)\\
    T_1(X_1),
  \end{array}\right .
\]
and let its equiprojectable decomposition for the order $Z<'Y <' X_1
<' X_2$ be given by triangular sets
\[
\Uu^{(i)} \left | 
  \begin{array}{l}
    U_{i,2}(Z,Y,X_1,X_2)\\
    U_{i,1}(Z,Y,X_1)\\
    S_{i}(Z,Y)\\
    R_i(Z),
  \end{array}\right .\qquad 1 \le i \le N.
\]
For $i\le N$, let $\tau_i: R_{\Uu^{(i)}} \to \K$ be the trace modulo
$\Uu^{(i)}$. Since $R_\Tt$ and $R_{\Tt'}$ are isomorphic
$\K$-algebras, the traces in $R_\Tt$ and $R_{\Tt'}$ coincide.  Since
$\langle \Tt'\rangle$ is the intersection of the pairwise coprime
ideals $\langle \Uu^{(i)}\rangle$, it follows (for instance from
Stickelberger's Theorem) that for any index $c$, we have
$${\rm tr}(A\, G^{c}) = \sum_{i \le N} \tau_i(Y Z^{c}).$$ For $i \le
N$, let $\ell_i$ be the linear form $\K[Z]/\langle R_i \rangle \to \K$
defined by $\ell_i(B)=\tau_i(YB)$. Then, one sees that $\tau_i(Y
Z^{c})=\ell_i(Z^{c})$, so that we have
\begin{equation} \label{eq:tAG}
{\rm tr}(A\, G^{c}) = \sum_{i \le N}\ell_i(Z^c).  
\end{equation}
Using
this remark, we can now give the whole algorithm and its running time.
\begin{itemize}
\item First, we compute $A \in R_\Tt$ such that $\ell =A \cdot {\rm
  tr}$. By Lemma~\ref{lemma:trdeg2}, this can be done in time
  $O(\M(d_1)\M(d_2)\log(d_1)\log(d_2))$.
\item Next, we compute the triangular sets ${\bf U}^{(i)}$,
  $i=1,\dots,N$.  This takes time $\EE(4,\delta_\Tt)$.

%% %% \footnote{Au
%% %%     final, pour le probl\`eme de projection des puissances, le
%% %%     $\EE(4,\delta_\Tt)$ n'appara\^it qu'une fois, contrairement \`a la
%% %%     composition modulaire, ou on a besoin de 2 changements
%% %%     d'ordre. Telegen permet de dire qu'on peut faire la composition
%% %%     modulaire avec un seul changement d'ordre ou c'est plus
%% %%     compliqu\'e que ca ?}
%% %% C'est un peu plus complique. L'algorithme met en jeu des calculs
%% %% non lineaires, avec des branchements. Je ne sais pas trop si 
%% %% ca se transpose.

\item The following step consists of computing the linear forms
  $\tau_i$ (by means of their values on the canonical bases of the
  residue class rings $R_{\Uu^{(i)}}$). We have seen in
  Subsection~\ref{ssec:basics} that we can compute each of those in
  time $O(\M(\delta_{\Uu^{(i)}}))$, so the total time is
  $O(\M(\delta_\Tt))$ by the super-linearity of $\M$.
\item Knowing the linear forms $\tau_i$, we can deduce $\ell_i$ by
  first computing all $Y\cdot \tau_i$ (for a total time of
  $O(\M(\delta_\Tt))$ again), from which the values of $\ell_i$ on the
  basis of $\K[Z]/\langle R_i \rangle$ can be read off.
\item Finally, we obtain ${\rm tr}(A\, G^{c})$, for
  $c=0,\dots,\delta_{\Tt}-1$, using Eq.~\eqref{eq:tAG} and the
  algorithm for transposed multiple reduction; this takes time
  $O(\M(\delta_\Tt) \log(\delta_\Tt)).$
\end{itemize}
Taking a quasi-linear $\M$, and summing all previous costs, the claim
in Theorem~\ref{theo:EtoC} follows.

\section*{Acknowledgments}

Adrien Poteaux is supported by the EXACTA grant of the National
Science Foundation of China (NSFC 60911130369) and the French National
Research Agency (ANR-09-BLAN-0371-01). \'Eric Schost is supported by
NSERC and the Canada Research Chair program. We wish to thank Marc
Moreno Maza and Yuzhen Xie for interesting discussions during the
preparation of this article.

\bibliographystyle{plain} \bibliography{equiproj}
\end{document}